\documentclass[conference]{IEEEtran}
\IEEEoverridecommandlockouts

\usepackage{cite}
\usepackage{amsmath,amssymb,amsfonts}
\usepackage{graphicx}
\usepackage{textcomp}
\usepackage{xcolor}

\usepackage{color}
\usepackage{mathtools}
\usepackage{xspace}
\usepackage{algorithmicx}
\usepackage{algpseudocode}
\usepackage{xcolor}
\usepackage{hyperref}
 \hypersetup{
     colorlinks,
      linkcolor={red!50!black},
     citecolor={blue!50!black},
     urlcolor={blue!80!black}
 }

\usepackage{nicematrix}
\usepackage{algpseudocode}
\usepackage{cases}

\usepackage{thmtools}

\usepackage{amsmath,amsthm}
\usepackage{cleveref}
\newtheorem{theorem}{Theorem}[section]
\newtheorem{lemma}[theorem]{Lemma}

\newtheorem{problem}[theorem]{Problem}
\newtheorem{proposition}[theorem]{Proposition}

\newtheorem{example}{Example}[theorem]

\newcommand{\set}[1]{\left\{#1\right\}}

\newcommand{\st}{\ :\ }
\newcommand{\nat}{\mathbb{N}}
\newcommand{\intg}{\mathbb{Z}}
\newcommand{\rel}{\mathbb{R}}
\newcommand{\rat}{\mathbb{Q}}
\newcommand{\com}{\mathbb{C}}
\newcommand{\alg}{\overline \rat}

\newcommand{\ii}{\mathbf{i}}

\newcommand{\sq}[1]{\langle #1 \rangle_{n\in\nat}}

\newcommand{\ie}{{\em i.e.}\xspace}
\newcommand{\eg}{{\em e.g.}\xspace}
\newcommand{\etc}{{\em etc.}\xspace}

\newcommand\defeq{\mathrel{\overset{\makebox[0pt]{\mbox{\normalfont\tiny\sffamily def}}}{=}}}

\newcommand\torus{\mathbb T}

\newcommand\multg[1]{\mathbb G^{#1}}

\renewcommand{\vec}[1]{\mathbf{#1}}
\newcommand{\vp}{\vec p}
\newcommand{\sset}[1]{\mathbf{#1}}
\newcommand{\orbit}[2]{\mathcal O_{#2}(#1)}

\newcommand{\conj}[1]{\overline{#1}}

\newcommand{\fr}[1]{\{#1\}_{2\pi}}

\newcommand{\modl}[1]{\left\|{#1}\right\|}

\def\BibTeX{{\rm B\kern-.05em{\sc i\kern-.025em b}\kern-.08em
    T\kern-.1667em\lower.7ex\hbox{E}\kern-.125emX}}
\begin{document}

\title{Multiple Reachability in Linear Dynamical Systems%
  \thanks{Toghrul Karimov and Jo\"el Ouaknine were supported by DFG grant
389792660 as part of TRR 248 (see
\url{https://perspicuous-computing.science}). Jo\"el Ouaknine was
supported by the European Research Council under Grant Agreement
101167561 (ERC Synergy Grant DynAMiCs), and is also
affiliated with Keble College, Oxford as \texttt{emmy.network} Fellow. James
Worrell was supported by EPSRC Fellowship EP/N008197/1.}
}

\author{\IEEEauthorblockN{Toghrul Karimov}
\IEEEauthorblockA{\textit{MPI for Software Systems} \\
  Saarbr\"ucken, Germany \\
  toghs@mpi-sws.org}
\and
\IEEEauthorblockN{Edon Kelmendi}
\IEEEauthorblockA{\textit{Queen Mary University of London} \\
London, United Kingdom \\
e.kelmendi@qmul.ac.uk}
\and
\IEEEauthorblockN{Jo\"el Ouaknine}
\IEEEauthorblockA{\textit{MPI for Software Systems} \\
Saarbr\"ucken, Germany \\
joel@mpi-sws.org}
\and
\IEEEauthorblockN{James Worrell}
\IEEEauthorblockA{\textit{University of Oxford} \\
Oxford, United Kingdom \\
jbw@cs.ox.ac.uk}
}

\maketitle

\begin{abstract}
  We consider reachability problems for linear dynamical systems.  In dimension $d$ these problems are specified by respective semialgebraic sets $\sset S,\sset T\subseteq \mathbb{R}^d$ of source and target states and a matrix $M\in \mathbb{Q}^{d\times d}$.  The task is to determine whether there is a point in $\sset S$ whose orbit under $M$ intersects the target $\sset T$ in at least $m$ distinct points.  The case $m=1$ (mere reachability) can be reduced to mild generalisations of the Skolem and Positivity Problems for linear recurrence sequences, whose decidability has been open for many decades.  The situation is markedly different for \emph{multiple reachability}, where $m$ can be greater than one.  In this paper, we prove that multiple reachability is undecidable already in dimension~$d=10$ with fixed multiplicity $m=9$.  Since our undecidability construction also shows that decision procedures for dimension $d \in \{3,\ldots,9\}$ would entail significant new results on effective solutions of Diophantine equations, we subsequently focus on the case $d=2$, that is, multiple reachability in the plane.  Here we obtain two positive results. We show that multiple reachability is decidable if the matrix $M$ is a rotation and it is also decidable without restriction on $M$ for halfplane targets.  The former result relies on a theorem in arithmetic geometry, due to Bombieri and Zannier, concerning intersections of algebraic subgroups with subvarieties.


\end{abstract}


\section{Introduction}
A \textbf{linear dynamical system} in ambient dimension $d$ is specified by a $d\times d$ matrix $M\in\rat^{d\times d}$ with rational entries. We are interested in understanding and deciding properties of the system's \textbf{orbit} for initial points $\vp\in\rat^d$, which is defined~as:
\begin{align*}
  \orbit\vp M \defeq \set{\vp\ M^n\st n\in\nat}.
\end{align*}
These are one of the simplest dynamical systems that we do not yet fully understand. They have been extensively studied for almost a hundred years. The motivations vary from finding solutions to Diophantine equations in number theory, to deciding linear loop termination in computer science, model checking simple programs \etc The text~\cite{recseq} is the principal introduction to linear dynamical systems, featuring the main theorems as well as a number of applications. See also~\cite{Karimov2022} for a recent survey. 

The core property we are interested in is \emph{reachability}: does the orbit reach some target set?  More precisely, a general phrasing of the \textbf{Reachability Problem} is the following. We are given respective source and target semialgebraic sets (defined by boolean combinations of polynomial inequalities) $\sset S, \sset T\subset \rel^d$, and a matrix $M\in \rat^{d\times d}$. The task is to decide whether there exists some point in the source set $\vp \in \sset S$, whose orbit under $M$ intersects the target set. In other words, does there exist $\vp \in \sset S$ such that
\begin{align*}
  \orbit \vp M \cap \sset T \ne\emptyset?
\end{align*}

A celebrated paper of Kannan and Lipton~\cite{kannan86_polyn_time_algor_orbit_probl} showed that point-to-point reachability (where both the source and target sets are singletons) is decidable in polynomial time, but for many variants of the Reachability Problem, decidability is open. Notably, point-to-hyperplane reachability (also known as Skolem's Problem) and point-to-halfspace reachability (also known as the Positivity Problem) have been studied extensively in relation to linear recurrence sequences, weighted automata, formal power series, model checking, and loop termination, but remain unsolved in general.  The current state of the art (see~\cite{shaullstacs}) is that the Reachability Problem is decidable in dimension $d=3$, Skolem's Problem is decidable in dimension $d=4$, and the Positivity Problem is decidable in dimension $d=5$.  In~\Cref{thm:reduction} we note that the Reachability Problem can be reduced to its point-to-polytope variant.  This last result suggests that the Skolem and Positivity Problems already capture much of the difficulty of the general (set-to-set) Reachability Problem.

In this paper we embark on a study of the \textbf{Multiple Reachability Problem}. This is a generalisation that does not merely ask whether the orbit intersects the target set, but rather whether it intersects it in at least $m$ points where $m\in\nat$ is part of the input. More precisely, we are given semialgebraic sets $\sset S, \sset T\subset \rel^d$, a matrix $M\in\rat^{d\times d}$, as well as a positive integer $m\in\nat$. The task is to decide if there is a point in the source set $\vp\in\sset S$ such that
\begin{align*}
  |\orbit \vp M \cap \sset T| \ge m.
\end{align*}

\begin{example}
  Here is a simple example:
  \begin{align*}
    \sset S &= \set{(x,y)\in\rel^2\st y=x^2},\\
    \sset T &= \set{(x,y)\in\rel^2\st x<y<-100},\\
    M&=\begin{pmatrix}
      2 & 0\\
      0 & -10/9
    \end{pmatrix},\text{ and } m = 5. 
  \end{align*}
  The answer to the multiple reachability problem for this instance is \emph{yes}. Since the linear map given by the matrix $M$ is particularly simple we can see the answer at once. Choose a point in $\sset S$ that is also in the second quadrant, \eg $\vp:=(-1,1)$. Observe that multiplication with $M^{2k+1}$, $k\in\nat$, sends $\vp$ to the fourth quadrant ($x<0$ and $y<0$), and the relation $x<y$ is invariant under this multiplication. Finally, from $\det M > 1$ it is clear that the orbit of $\vp$ under $M$ enters the target set $\sset T$ at least $m=5$ times. Indeed it enters the target infinitely often. 
\end{example}

What is the difference between the Reachability and Multiple Reachability problems? Our first observation is that, surprisingly, the Multiple Reachability Problem is computationally much more difficult than (single) Reachability.

\subsection{Contributions}
\begin{theorem}
  \label{thm:undecidable}
  The Multiple Reachability Problem is undecidable in general and is already undecidable in dimension $d=10$ with multiplicity $m=9$.
\end{theorem}
This is in stark contrast to the Reachability Problem---no natural variants of which are known to be undecidable and which, as remarked above, can be reduced to its point-to-polytope variant.

Intuitively, the lack of natural undecidable variants for reachability is because there is a single deterministic  rule that governs the dynamics of the system. In other words, these are programs without conditionals. In dynamical systems which have some non-determinism, \ie when the dynamics is governed by at least two maps, undecidable problems abound. For example, emptiness of probabilistic automata~\cite{gimbert-PA} can be seen as a point-to-halfspace reachability problem, but where we have at least two linear maps $M_1$, $M_2$ at our disposal, to move the point to the target. The choice between the two dynamics is used to simulate a Turing machine. We have to proceed differently for the proof of \Cref{thm:undecidable}. We reduce from a variant of Hilbert's tenth problem. The instances are encoded in the source set $\sset S\subseteq \rel^d$, so that points $\vp \in \sset S$ contain some \emph{real} zero of the given polynomial. Afterwards, the matrix $M$ is constructed in such a way that the orbit of $\vp$ under $M$ reaches some hyperplane if and only if the coordinates of $\vp$ are distinct natural numbers. This last step is made possible by the fact that every univariate polynomial of degree $d$ satisfies the same linear recurrence relation. In the reduction the matrix $M$ is not diagonalisable, and the proof would not work if we restricted $M$ to be diagonalisable. 

Hilbert's tenth problem is undecidable for $9$ variables, and consequently our reduction implies that multiple reachability with algebraic initial and hyperplane target sets is undecidable in dimension $d=19$ for fixed $m=9$.
Similarly, for semialgebraic initial and hyperplane target sets undecidability follows in dimension $d=10$. 
More generally, decidability of the Multiple Reachability Problem in dimension $d$ would give us algorithms to solve Diophantine equations in $d-1$ variables, which is open and considered very difficult already for $d=3$.
For $d\ge4$, it is conceivable that whether a solution exists might even be undecidable. Indeed, effectively solving Thue equations (homogeneous equations in two variables) was only possible after Baker's work on linear forms in logarithms in 1966; See, for example, \cite{waldschmidt2020thue}. 
Therefore, we focus our search for positive results on the two-dimensional affine plane~$\rel^2$. Here we show: 
\begin{theorem}
  \label{thm:pos1}
  In dimension $d=2$ the Multiple Reachability Problem is decidable
  (i)~when $\sset T$ is a halfspace (with $\sset S$ and $M$ arbitrary) or
  (ii)~when $M$ is a rotation (with $\sset S$ and $\sset T$ arbitrary).
\end{theorem}

\Cref{thm:pos1}(i) is proved using Kronecker's Theorem on Diophantine approximation and quantifier-elimination for the first-order theory of real-closed fields. \Cref{thm:pos1}(ii), is the main contribution of the present paper.

Most decidability results about linear dynamical systems are proved using Baker's effective bounds on linear forms in logarithms.  For the proof of \Cref{thm:pos1}(ii), we make crucial use of bounds, due to Bombieri and Zannier, on the height of algebraic points in the set of intersections between a variety and algebraic subgroups of low dimension.  To the best of our knowledge this is the first use of such tools in the analysis of linear dynamical systems, and it is intriguing that they are apparently needed to handle even special cases of multiple reachability in the plane.  The general case of the Multiple Reachability Problem in the plane remains open.

\subsection{\Cref{thm:pos1}(ii) Proof and Algorithm Overview}

We reduce to the following natural problem: Given a semialgebraic\footnote{Here we mean that the image of $T$ under the map $f \colon \com^k \to \rel^{2k}$ that extracts real and imaginary parts of coordinates is semialgebraic.} set $\sset T\subset\com^k$, and an algebraic number $\lambda$ with $|\lambda|=1$, decide whether the intersection of
\begin{align}
  \label{eq:intersection}
  \set{(\lambda^{x_1},\ldots, \lambda^{x_k})\st x_1\ldots, x_k\text{ distinct positive integers}}
\end{align}
with $\sset T$ is empty. To prove that this problem is decidable, we give a procedure for solving systems of polynomial (in)equalities in powers of an algebraic number $\lambda$ which is in the unit circle.

Every point in the set of powers of $\lambda$ in \eqref{eq:intersection} belongs to an algebraic group of dimension~1. Algebraic groups are algebraic sets, \ie solutions to a system of polynomial equations, that have a group structure. In our case the group operation is component-wise multiplication. Intuitively, the dimension is~1 because we have only one algebraic number $\lambda$ and it lies on the unit circle. On the other hand, the semialgebraic set $\sset T$ can be assumed to be the intersection of an algebraic set, or variety, $X$, and another open semialgebraic set that is specified as an intersection of strict polynomial inequalities. 

There are a number of conjectures and results, related to the Mordell-Lang Conjecture, that roughly say: if the intersection between an algebraic group of low dimension and a variety is \emph{large} then there must be some simple algebraic reason. See the book~\cite{zannier2012some} for an overview. The salient result for our purposes, due to Bombieri and Zannier, can be found in the appendix of~\cite{schinzel2000polynomials}. This theorem says that there is a partition of any variety $X$ into $X=X^\circ\cup X^\bullet$ such that the intersection of $X^\circ$ with the union of all groups of dimension~1 has bounded Weil height; moreover, inspecting the proof, one sees that the bound is effective. This upper bound directly translates to a bound on $x_i$ in the intersection \eqref{eq:intersection}. Further results by Bombieri, Schmidt, Zannier and others are used for computing the defining equations of the set~$X^\bullet$, which contains all solutions that are degenerate in some sense.

The algorithm computes the description of the subset $X^\bullet$ as the set of common zeros of finitely many polynomials, as well as a bound on the exponents $x_i$ of $\lambda$. 
It then checks finitely many tuples $(x_1,\ldots,x_k)$ to see whether they form a solution.
These tests use Tarski's algorithm for quantifier elimination in real closed fields as a subroutine. 

\subsection{Example}
The following while loop is a slightly more complex example, highlighting the connection to program analysis.\\

\begin{minipage}{0.5\textwidth}
    \begin{algorithmic}
    \State {\textbf{assume} $x^3+xy^2=2y^2$}
    \State {$m\gets 2$}
    \While {$m\ne 0$}\\
    \ \ \
\State{$\begin{pmatrix} x\\ y\end{pmatrix} \gets 
    \begin{pmatrix}
    4/5 & 3/5\\
    -3/5 & 4/5 
  \end{pmatrix}
  \begin{pmatrix} x\\ y\end{pmatrix}$}\\
     \If {$x=y+1$}
        \State $m\gets m-1$
      \EndIf
    \EndWhile
  \end{algorithmic}
  \vspace{0.2cm}
\end{minipage}
\begin{minipage}{0.47\textwidth}
    \fbox{\includegraphics[width=\textwidth]{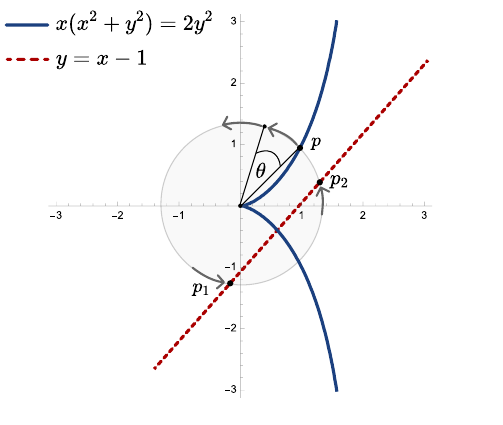}}
  \vspace{0.2cm}
\end{minipage}

Does this program terminate? More precisely, is there some initialisation of the variables $x,y\in\rel$ such that they satisfy the polynomial relation\footnote{This curve is the \emph{cissoid of Diocles}, discovered around 180 BC. See~\cite[Chapter~15]{lockwood1967book}.}  \begin{align}
  \label{eq:ex pol}
  x^3+xy^2=2y^2,
\end{align}
and for which the program terminates? Let us reinterpret this question as follows. First we notice that the vector $(x,y)$ is being updated with the matrix
\begin{align*}
  \begin{pmatrix}
    4/5 & 3/5\\
    -3/5 & 4/5 
  \end{pmatrix},
\end{align*}
which has the property that for all $n\in\nat$ and $\theta=-\cos^{-1}(4/5)$:
\begin{align*}
  \begin{pmatrix}
    4/5 & 3/5\\
    -3/5 & 4/5
  \end{pmatrix}^n=  \begin{pmatrix}
    \cos n\theta & -\sin n\theta\\
    \sin n\theta & \cos n\theta
  \end{pmatrix}.
\end{align*}
We see that with every loop iteration, the updates rotate the point $(x,y)$ by the angle $\theta$ on the affine plane. So the question of the termination of the program above is the question of whether there is some point $\vp$ in the cissoid defined in~\eqref{eq:ex pol}, that can be rotated into at least two points of the line $y=x-1$.

The algorithms  we present in this paper can be used to answer questions like the above (and more). In this example, the answer turns out to be negative; there are no points in the cissoid that can be rotated by $\theta$ to two different points on the line. Therefore, if the variables $x,y$ are initialised such that they satisfy the polynomial \eqref{eq:ex pol}, the while loop, above, does not terminate. 

\subsection{Related Work}
Effective procedures for reachability in linear dynamical systems have been investigated for a long time. There are various partial results. A brief survey of the state of the art can be found in~\cite{Karimov2022}.

Directly related to the present paper, the semialgebraic-to-semi\-algebraic (single) reachability problem was studied in~\cite{shaullstacs}. There, this decision problem is shown decidable when the dimension is $3$, using Baker's effective estimates. Furthermore, \cite{shaullstacs} shows by way of hardness that an algorithm for deciding this problem in dimension $4$ would entail the ability to effectively estimate Lagrange constants of certain transcendental numbers. The proof of \Cref{thm:reduction} appears implicitly in~\cite[Theorem~11]{shaullstacs}.

More closely related to \emph{multiple} reachability is the question of multiplicity in linear recurrence sequences. A consequence of the Skolem-Mahler-Lech theorem is that for any integer $k$, and any non-degenerate linear recurrence sequence $\sq{u_n}$, there are only finitely many $n$ for which $u_n=k$. Thus one can ask what is the largest  number
solutions $n$ of such an equation when $\sq{u_n}$ ranges over non-degenerate linear recurrence sequences of a certain order. Equivalently, what is the largest number of times a non-degenerate linear dynamical system from a singleton source hits a hyperplane target? There are many interesting and deep answers to this question, see \cite[Chapter 2.2]{recseq} and references therein. 

The questions that we consider in this paper are generalisations of the Skolem Problem. There is another interesting generalisation in a different direction, which happens to be undecidable for nontrivial reasons. Namely, given $k$ linear recurrence sequences over algebraic numbers
\begin{align*}
  \sq{u_n^{(1)}}, \sq{u_n^{(2)}},\ldots, \sq{u_n^{(k)}}, 
\end{align*}
we are asked to decide whether there are natural numbers $n_1,\ldots, n_k$ such that
\begin{align*}
  u^{(1)}_{n_1}+u_{n_2}^{(2)}+\cdots+u_{n_k}^{(k)}=0. 
\end{align*}
This problem was conjectured to be undecidable by Cerlienco, Mignotte, and Piras in \cite{cerlienco1987linear}. The conjecture was proved by Derksen and Masser a few years ago in \cite{derksenmasser15}, for $k=557844$. Similarly to the present paper, they reduce from Hilbert's tenth problem, and their proof requires that the sequences not be diagonalisable.

\section{Definitions and Basic Properties}
We define the natural numbers as the set $\nat = \set{1,2,3,\ldots}$. Atomic formulas of the \textbf{first-order logic of reals} are propositions of the type:
\begin{align*}
  P(x_1,\ldots,x_n)>0,
\end{align*}
where $x_1,\ldots, x_n$ are first-order variables ranging over $\rel$, and $P\in\intg[x_1,\ldots,x_n]$ is a polynomial with integer coefficients. Atomic propositions can be combined with Boolean connectives, and we can also quantify over the set of real numbers. This logic admits effective quantifier elimination via Tarski's algorithm~\cite{tarski1951decision}. This means that given a formula:
\begin{align*}
  \exists x_0\ \Phi(x_0,x_1,\ldots,x_n),
\end{align*}
there is an equivalent quantifier-free formula $\Gamma(x_1,\ldots,x_n)$ that can be effectively computed. In particular, given a sentence (\ie a formula with no free variables), Tarski's procedure can be used to decide whether the sentence is true over $\mathbb R$.

Subsets $\sset S\subseteq\rel^d$ that can be expressed using formulas in the logic described above, that is
\begin{align*}
  \sset S=\set{(x_1,\ldots, x_d)\in\rel^d\st \Phi(x_1,\ldots,x_d)},
\end{align*}
for some formula $\Phi$, are called \textbf{semialgebraic}. Due to quantifier elimination, semialgebraic sets are exactly the sets $\sset S\subseteq \rel^d$ that can be written as finite unions of sets of tuples $(x_1,\ldots,x_d)\in\rel^d$ that satisfy simultaneously
\begin{align}
  \label{eq:sa sets}
  \begin{cases}
    P_0(x_1,\ldots,x_d)&=0,\\
    P_1(x_1,\ldots,x_d)&>0,\\
    \vdots\\
    P_k(x_1,\ldots,x_d)&>0,
  \end{cases}
\end{align}
where $P_i\in\intg[x_1,\ldots,x_d]$. We only need one equality because the intersection of real zeros of polynomials $P$ and $Q$ is exactly the set of real zeros of the polynomial $P^2+Q^2$. In this setting, an \textbf{algebraic set} is the set of zeros of a polynomial with integer coefficients. A \textbf{hyperplane} is the set of solutions of an affine equation, \ie $(x_1,\ldots, x_d)\in\rel^d$ for which
\begin{align*}
  a_1x_1+\cdots+a_dx_d+a_{d+1}=0,
\end{align*}
where $a_i$ are integers. A \textbf{halfspace} is the set of solutions of an affine \emph{inequality}, and a \textbf{polytope} is the intersection of finitely many halfspaces. On~$\rel^2$, a hyperplane is just a \textbf{line}, and a halfspace is called a \textbf{halfplane}.
Finally, when discussing semialgebraicity for subsets of $\com^d$, we identify the latter with $\rel^{2d}$ by taking real and imaginary parts.

A \textbf{linear recurrence sequence} is a sequence $\sq{u_n}$ of rational numbers that satisfies a linear recurrence relation
\begin{align}
  \label{eq:lrs}
  u_n=a_1u_{n-1}+\cdots + a_du_{n-d},
\end{align}
for all $n>d$, where $a_i$ are rational numbers. The smallest positive number $d$ for which the sequence satisfies \eqref{eq:lrs} is called the \textbf{order} of the sequence.
A \textbf{linear dynamical system}  evolves according to the map $x \mapsto Mx$ for $M \in \rat^{d\times d}$.
Linear recurrence sequences and linear dynamical systems are essentially the same object, as summarised in the two following propositions.
\begin{proposition}
  \label{prop:lrs to mat}
  Let $\sq{u_n}$ be a linear recurrence sequence of order $d$. Then there exists $M\in\rat^{d\times d}$ such that
  \begin{align*}
    u_n=(M^n)_{1,d}\,\,\, \text{for all $n\in\nat$.}
  \end{align*}
\end{proposition}
\begin{proposition}
  \label{prop:mat to lrs}
  Let $M\in\rat^{d\times d}$ and $1\le i,j\le d$. Then
  \begin{align*}
    \sq{(M^n)_{i,j}}
  \end{align*}
  is a linear recurrence sequence of order at most $d$. 
\end{proposition}
\noindent The proof of \Cref{prop:lrs to mat} is elementary, and \Cref{prop:mat to lrs} follows from the Cayley-Hamilton theorem; See \cite[Chapter 1]{recseq} for more details. Furthermore, both propositions are effective. 

The \textbf{characteristic polynomial} of the linear recurrence \eqref{eq:lrs} is
\begin{align*}
  x^d-a_1x^{d-1}-a_2x^{d-2}-\cdots -a_d.
\end{align*}
Denote by $\Lambda_1,\ldots ,\Lambda_k$ the distinct roots of this polynomial and by $m_1,\ldots,m_k$ their respective multiplicities. A linear recurrence sequence $\sq{u_n}$ can also be written as a \textbf{generalized power sum}, which is an expression of the form
\begin{align*}
  u_n=\sum_{i=1}^kP_i(n)\ \Lambda_i^n,
\end{align*}
where $P_i\in\alg[n]$ are polynomials of degree at most $m_i-1$. Furthermore, all generalized power sums satisfy linear recurrence relations with algebraic coefficients. A consequence of this fact is that linear recurrence sequences are closed under addition and product: if $\sq{u_n}$ and $\sq{v_n}$ are two linear recurrence sequences, then so are the sequences $\sq{u_n+v_n}$ and $\sq{u_n\cdot v_n}$.

These are all the necessary facts required to prove the following:
\begin{theorem}
  \label{thm:reduction}
  The general Reachability Problem reduces to the point-to-polytope variant. 
\end{theorem}

The main idea appears implicitly in the proof of~\cite[Theorem~11]{shaullstacs}.

\begin{proof}
  Suppose that we are given an instance of the semialgebraic to semialgebraic reachability problem. Let $d\in\nat$ be the dimension of its ambient space, $\sset S,\sset T\subseteq \rel^d$ be the source and target sets respectively, and $M$ be the given matrix. Denote by $\Phi_{\sset S}$, $\Phi_{\sset T}$, the formulas  defining the respective sets $\sset S, \sset T$. Write $\vec x$ for the tuple of variables $(x_1,\ldots, x_d)$ and $A$ for the $d \times d$ matrix of variables $(A_{1,1},\ldots,A_{d,d})$, and define the formula
  \begin{align*}
    \Gamma(\vec x, A)\defeq \Phi_{\sset S}(\vec x) \wedge \Phi_{\sset T}(\vec x \cdot A).
  \end{align*}
  The reachability problem asks whether there exists $\vp\in\rel^d$ and $n\in\nat$ such that $\Gamma(\vp, M^n)$ holds. Since the first-order theory of reals admits effective quantifier elimination, we first use Tarski's algorithm to produce a quantifier-free formula $\Gamma'(A)$, which is equivalent to the projection $\exists \vec x\ \Gamma(\vec x, A)$. Now the reachability problem is equivalent to the question of whether there exists some $n$ such that $\Gamma'(M^n)$ holds.  Since $\Gamma'$ is quantifier-free, it can be written as a disjunction of formulas $\varphi_1,\ldots, \varphi_m$, for some $m\in\nat$, such that each $\varphi_i$ is of the form \eqref{eq:sa sets}. 
  It suffices to construct, for each $\varphi_i$, an instance of the point-to-polytope reachability problem with the property that $\varphi_i(M^n)$ holds for some $n$ if and only if the respective polytope is reached from some point in~$\sset{S}$. 
  We can then take the union of these polytopes as the single polytopic target.
  Let $\varphi$ be one of the disjuncts, written in the form
  \begin{align*}\bigwedge
    \begin{cases}
      P_0(A_{1,1},\ldots,A_{d,d})&=0,\\
      P_1(A_{1,1},\ldots,A_{d,d})&>0,\\
      \vdots\\
      P_k(A_{1,1},\ldots,A_{d,d})&>0.
    \end{cases}
  \end{align*}
  Define for $i\in\set{0,\ldots,k}$ the sequences
  \begin{align*}
    u_{i,n}\defeq P_i\left ( (M^n)_{1,1},\ldots,(M^n)_{d,d} \right ) ,\ n\in\nat. 
  \end{align*}
  It follows from \Cref{prop:mat to lrs} and the closure of linear recurrence sequences under component-wise addition and multiplication, that all the sequences $\sq{u_{i,n}}$ are themselves linear recurrence sequences.
  Write $d_i$ for the order of $\sq{u_{i,n}}$.
  Applying \Cref{prop:lrs to mat} we construct matrices $N_i$ of size $d_i\times d_i$ for $0\le i\le k$, with the property that the upper-right corner of $N_i^n$ is equal to $u_{i,n}$.

  Unravelling the definitions, we see that for all $n\in\nat$, $\varphi(M^n)$ holds if and only if the upper-right corner of $N_0^n$ is $0$, and the upper-right corners of $N_i^n$, $1\le i\le k$ are strictly positive. The latter can be interpreted as a point-to-polytope reachability problem as follows. Let $D:=\sum d_i$, and construct a block diagonal matrix whose blocks are $N_0,\ldots, N_k$, and whose size is $D\times D$. Then the equivalent instance of the point-to-polytope problem has as initial point $\vp_0:=(1,\ldots,1)\in\rel^D$, the matrix is $N$ and the polytope is the intersection of the following halfspaces. The closed halfspaces are characterised by the normal vectors $\Delta(d_0)$ and $-\Delta(d_0)$ (where by $\Delta(i)\in\rel^D$ we denote the vector whose components are all zero except the component in position $i$ whose value is $1$), and the open halfspaces with normal vectors $\Delta(d_1),\ldots, \Delta(d_k)$. 
\end{proof}
Why does a similar proof not work for multiple reachability? The critical difference occurs after we obtain the projection $\Gamma'$. If there are two distinct integers $n_1,n_2$ such that $\Gamma'(M^{n_1})$ and $\Gamma'(M^{n_2})$ hold, it does not necessarily mean that there is a \emph{single} $\vp$ for which both $\Gamma(\vp,M^{n_1})$ and $\Gamma(\vp,M^{n_2})$ hold. Indeed, it is unlikely that such a reduction is possible for multiple reachability, in light of the result of the next section.

\section{Hilbert's Tenth Problem and Linear Dynamical Systems}
In this section we prove the undecidability of the multiple reachability problem, with algebraic starting sets and hyperplane targets, by reducing from a variant of Hilbert's tenth problem.\footnote{A sketch of this proof has already appeared in \cite{Karimov2022}.} The variant that we reduce from is the following:

\begin{problem}
  \label{prob:hilbert}
  Given a polynomial $P(x_1,\ldots, x_k)$ with integer coefficients, decide whether there are distinct positive integers $n_1,n_2,\ldots, n_k$ such that
  \begin{align*}
    P(n_1,\ldots, n_k)=0.
  \end{align*}
\end{problem}
\begin{proposition}
  \label{prop:hilbert}
  \Cref{prob:hilbert} is undecidable. 
\end{proposition}
\begin{proof}
  Let $Q(x_1,\ldots, x_n)$ be an arbitrary polynomial with integer coefficients. For any partition $\mathcal P$ of $\set{1,\ldots,n}$, define $Q_{\mathcal P}$ to be the polynomial obtained by taking $Q$ and for every $A\in \mathcal P$ replacing all variables $x_i$, for $i\in A$, by a single fresh variable. Clearly $Q$ has a zero in positive integers $x_1,\ldots, x_n$ if and only if one of the polynomials $Q_{\mathcal P}$ has a zero in \emph{distinct} positive integers. Since Hilbert's tenth problem is undecidable (\ie there is no procedure that can decide whether a given polynomial has a zero in positive integers, see \cite[Chapter~5]{matijasevic}), it follows that \Cref{prob:hilbert} is also undecidable. 
\end{proof}
Hilbert's tenth problem is known to be undecidable even when the number of variables is fixed, equal to $9$. As a consequence of the proof above we have the following corollary.
\begin{proposition}[\cite{jones}]
  \label{prop:hilbert2}
  \Cref{prob:hilbert} is undecidable for fixed $k=9$. 
\end{proposition}

We will now show that \Cref{prob:hilbert} can be reduced to the multiple reachability problem. This comprises two steps. First we prove that all univariate polynomials of degree $d$ satisfy the same linear recurrence relation, which is then turned into a matrix form. In the second step we construct a certain algebraic set from the polynomial of \Cref{prob:hilbert}.

\begin{lemma}
  \label{cor:recurrence}
  Let $P$ be a univariate polynomial of degree $d$. The unique sequence
$\sq{v_n}$
  that satisfies the recurrence
  \begin{align}
    \label{eq:recurrence}
    \sum_{i=0}^{d+1}(-1)^i{d+1\choose i}v_{n-i}=0,\qquad n>d+1. 
  \end{align}
  and whose first $d+1$ entries are $P(1),P(2),\ldots,P(d+1)$ is the sequence
  \begin{align*}
    \sq{P(n)}.
  \end{align*}
\end{lemma}
\begin{proof}
  Let $P$ be a univariate polynomial of degree $d$. We prove that the sequence $\sq{P(n)}$ satisfies the recurrence relation~\eqref{eq:recurrence}. This suffices because any sequence satisfying a recurrence relation of order $d+1$ is determined by its first $d+1$ terms.

  Define the discrete difference operator $\Delta\st \intg[x]\to\intg[x]$ by
  \begin{align*}
    (\Delta f)(x)\defeq f(x)-f(x+1).
  \end{align*}
  Note that $\Delta$ is linear and that $\Delta f$ has degree at most $\mathrm{deg}(f)-1$. It follows that $\Delta^{d+1}P=0$, since $P$ has degree $d$. We establish the identity
  \begin{align}
    \label{eq:delta identity}
    (\Delta^k P)(x)=\sum_{i=0}^k(-1)^i{k\choose i}P(x+i),
  \end{align}
  by induction on $k$. The base case is evident, for the induction step we proceed as follows. Using the linearity of $\Delta$ we have:
  \begin{align*}
    (\Delta^{k+1}P)(x) = \sum_{i=0}^k (-1)^i{k\choose i}\left(P(x+i)-P(x+i+1)\right)
  \end{align*}
  Taking out the first and last term, splitting the sum and shifting the index by one, it is possible to write the right hand side of the equation above as:
  \begin{align*}
    &P(x)+(-1)^{k+1}P(x+k+1)\\
    &+\sum_{i=1}^k(-1)^{i}\left({k\choose i} + {k\choose i+1}\right)P(x+i).
  \end{align*}
  Using a binomial identity now we can finish the proof by writing the above as
  \begin{align*}
    &P(x)+(-1)^{k+1}P(x+k+1)+\sum_{i=1}^k(-1)^{i}{k+1\choose i}P(x+i)\\
    &=\sum_{i=0}^{k+1}(-1)^{i}{k+1\choose i}P(x+i). 
  \end{align*}
  Thus we have the identity \eqref{eq:delta identity}, which when instantiated for $k=d+1$ proves that the sequence $\sq{P(n)}$ satisfies the recurrence relation~\eqref{eq:recurrence}. 
 \end{proof}

Let us turn the statement of the above lemma into matrix form. To this end let $d\in\nat$ be a natural number. Denote the $d+1$ coefficients of the recurrence \eqref{eq:recurrence} by
\begin{align*}
  q_{i}\defeq (-1)^{i+1}{k+1\choose i},\qquad 1\le i\le d+1. 
\end{align*}
Let $\vec h_d:=(1,0,\ldots,0)\in\rel^{d+1}$ and define the matrix
\begin{align*}
  M_d\defeq \begin{pNiceMatrix}[margin]
    0 & 0 &  \cdots & 0 & q_{d+1}\\
   \Block[fill=red!15,rounded-corners]{4-4}{}
    1 & 0 &  \cdots & 0 & q_d\\
    0 & 1 &  \cdots & 0 & q_{d-1}\\
    \vdots & \vdots &\ddots & \vdots & \vdots\\
    0 & 0 & \cdots & 1 & q_1
  \end{pNiceMatrix},
\end{align*}
where the shaded block is the $d\times d$ identity matrix. 
It follows from the discussion above that for all univariate polynomials $P$ of degree $d$, and all $n\in\nat$, we have
\begin{align}
  \label{eq:equal poly}
  \big (P(1),P(2),\ldots,P(d+1)\big)\ M_d^n\ \vec h_d^\top=P(n).
\end{align}

To reduce Problem~\ref{prob:hilbert} to the algebraic-to-hyperplane multiple reachability, we proceed as follows. Let $F\in\intg[y_1,\ldots,y_n]$ be an arbitrary polynomial with integer coefficients. We define an algebraic set $S\subseteq \rel^{2n+1}$ as: 
\begin{align*}
  &(x_1,\ldots,x_{n+1},y_1,\ldots,y_n)\in S \Leftrightarrow\\
  &\bigwedge \begin{cases}
    F(y_1,\ldots,y_n)=0,\\
    x_1=(1-y_1)(1-y_2)\cdots (1-y_n),\\
    x_2=(2-y_1)(2-y_2)\cdots (2-y_n),\\
    \vdots\\
    x_{n+1}=(n+1-y_1)(n+1-y_2)\cdots (n+1-y_n).
  \end{cases}
\end{align*}
The idea is that to check whether a root $(y_1,\ldots,y_n)$ of $F$ is in $\nat^n$, we need only check that the sequence $(m-y_1)\cdots (m-y_n)$, $m\in\nat$, has $n$ zeros. More precisely, denote by $M$ the $(2n+1)\times (2n+1)$ matrix whose first $(n+1)\times (n+1)$ block is equal to $M_{n}$ and the other entries are $0$, and set $\vec h := \vec h_{2n}$. 
\begin{lemma}
  \label{lem:hilbert equivalence}
The following two statements are equivalent: 
\begin{itemize}
\item The polynomial $F$ has a root consisting of distinct positive integers. 
\item There is some $\vec p:=(x_1,\ldots,x_{n+1},y_1,\ldots,y_n)\in S$ and distinct positive integers $r_1,\ldots,r_n$ such that
  \begin{align*}
    \vec p\ M^{r_i}\ \vec h^\top = 0,\qquad 1\le i\le n. 
  \end{align*}
\end{itemize}
\end{lemma}
\begin{proof}
  \noindent($\Rightarrow$) Let $y_1,\ldots, y_n$ be distinct positive integers with $F(y_1,\ldots,y_n)=0$.
  For all $i\in\set{1,\ldots, n+1}$, set
  \begin{align*}
    x_i:=(i-y_1)(i-y_2)\cdots (i-y_n),
  \end{align*}
  Then $\vec p:=(x_1,\ldots,x_{n+1},y_1,\ldots, y_n)\in S$ by definition. The definition of the matrix $M$ above (that has nonzero entries only in the first $(n+1)\times (n+1)$ block) and \eqref{eq:equal poly} imply that for all $r\in\nat$ we have
  \begin{align}
    \label{eq:equal poly2}
    \vp\ M^r\ \vec h^\top=(r-y_1)(r-y_2)\cdots (r-y_n). 
  \end{align}
  Hence the second statement of the lemma holds for the distinct positive integers $r_i=y_i$. 

  \noindent($\Leftarrow$) Let $\vp$ and distinct positive integers $r_1,\ldots, r_n$ be such that the second statement holds. Then \eqref{eq:equal poly2} implies that the tuple $(y_1,\ldots,y_n)$ is a permutation of the tuple of distinct positive integers $(r_1,\ldots, r_n)$. It follows from the definition of $S$ that the same permutation is also a root~of~$F$. 
\end{proof}
\Cref{prop:hilbert} and \Cref{lem:hilbert equivalence} imply that al\-ge\-bra\-ic-to-hy\-per\-plane multiple reachability is undecidable, which is~\Cref{thm:undecidable}. Indeed the set $S$ defined above is algebraic,\footnote{As mentioned in the previous section, the real vectors $\vec x$ for which $P(\vec x)=0$ and $Q(\vec x)=0$ coincide with the real vectors $\vec x$ for which $P(\vec x)^2+Q(\vec x)^2=0$.} and $\vec h$ is the normal vector of some hyperplane (recall that a point $\vec x$ is on the hyperplane with a normal vector $\vec h$ if and only if $\vec x\cdot\vec h^\top=0$).

More precisely, we have shown that a procedure to decide al\-ge\-bra\-ic-to-hyperplane multiple reachability in dimension $2n+1$ can be used to effectively solve Diophantine equations with $n$ variables. By projecting away the coordinates $y_1,\ldots, y_n$ in the definition of $S$ above,  we obtain a semialgebraic set. Hence a procedure to decide \emph{semialgebraic}-to-hyperplane multiple reachability in dimension $n+1$ can be used to effectively solve Diophantine equations with $n$ variables. In light of \Cref{prop:hilbert2}, we have the following theorem.
\begin{theorem}
  \label{thm:undecidable dim}
  Algebraic-to-hyperplane multiple reachability is undecidable in dimension~$19$.
Semialgebraic-to-hyperplane multiple reachability is undecidable in dimension~$10$.
\end{theorem}
Effectively solving Diophantine equations is notoriously difficult. Even Thue equations, \ie equations of the type $P(\vec x)=m$ where $P$ is a homogeneous polynomial, could only be solved effectively in the second half of the twentieth century, after the work of Alan Baker~\cite[Theorem 4.1]{baker1990transcendental}. As a consequence, in the next section, we focus our efforts in understanding the multiple reachability problem on the affine plane, \ie when the dimension is fixed at $d=2$. As we shall see, even on the plane, multiple reachability can be quite challenging.

In the undecidability proof of this section, the matrix $M$ is not diagonalisable. It is interesting to explore the multiple reachability problem for diagonalisable matrices, as the latter is a property that holds for generic matrices. This is at least as hard as the Positivity Problem for diagonalisable linear recurrence sequences.

\section{Algorithms on the Affine Plane}
This section is devoted to the proof of \Cref{thm:pos1}. The dimension $d=2$ is fixed. The system is given in the form of a $2 \times 2$ matrix with rational entries. The eigenvalues of such a matrix have one of the following forms: (a)~a~pair of complex conjugates $\lambda, \conj\lambda\in\alg$, (b) two real roots $\rho_1, \rho_2\in\alg\cap\rel$, or (c) a repeated real root $\rho\in\alg\cap\rel$. When the eigenvalues are a pair of complex conjugates and $|\lambda|=1$ we say that the matrix is a \textbf{rotation}.

\Cref{thm:pos1} consists of two statements. The first statement, \Cref{thm:pos1}(i), restricts targets to halfspaces (whereas the matrix is arbitrary), and its proof is postponed to \Cref{sec:halfplane targets}. The second statement, \Cref{thm:pos1}(ii), restricts the matrix to rotations (whereas the target is arbitrary); its proof follows.

\subsection{Rotations}
We will first give the simple reduction to solving systems of polynomial inequalities in powers of some $\lambda$, as discussed in the introduction, followed by a proof overview.

\subsubsection{Reduction}
\label{sec:reduction}

Let $\vec S,\vec T\subseteq \rel^2$ be the source and target semialgebraic sets, given by the formulas $\Phi_{\vec S}, \Phi_{\vec T}$ of first-order logic of reals.
Further, let $M$ be a matrix whose eigenvalues are the pair $\lambda,\conj\lambda$ on the unit circle, that is $|\lambda|=1$, and let $m\in\nat$. 
We have to give a procedure for deciding whether there exists some $\vp\in\vec S$ and distinct positive integers $x_1,\ldots, x_m\in\nat$ such that
\begin{align*}
  \vp\ M^{x_i}\in \vec T, 
\end{align*}
for all $i\in\set{1,2,\ldots, m}$.

We proceed by eliminating the existential quantifier in the decision question. To this end, let $\vec v =(v_1,v_2)$ be a tuple of variables, let $V_1,\ldots, V_m$ be $2\times 2$ matrices of fresh variables, and consider the following formula:
\begin{align*}
  \Gamma(\vec v, V_1,\ldots, V_m)\defeq \Phi_{\vec S}(\vec v) \land \bigwedge_{i=1}^m\Phi_T\left(\vec v\ V_i\right). 
\end{align*}
The multiple reachability decision problem asks whether there is some $\vp\in\rel^2$ and distinct positive integers $x_1,\ldots, x_m$ such that
\begin{align}
  \label{eq:dec prob}
  \Gamma(\vp, M^{x_1},\ldots, M^{x_m})
\end{align}
holds. Eliminating the existential quantifiers for $\vec v$ from $\Gamma$, we effectively obtain another formula $\Gamma'(V_1,\ldots, V_m)$  such that \eqref{eq:dec prob} holds for some point $\vec p$ if and only if $\Gamma'(M^{x_1},\ldots, M^{x_m})$ is true. Tuples of reals that satisfy $\Gamma'$ form a semialgebraic set; which can be written as a finite union of sets of the form \eqref{eq:sa sets}, that is a system of one polynomial equality and a finite number of polynomial inequalities. Each set in this union can be treated separately, so let $P_0,\ldots, P_{\ell}$ be polynomials (with integer coefficients) of one of the sets:
\begin{align*}
  \Psi(V_1,\ldots, V_m)\defeq
  \bigwedge\begin{cases}
    P_0(V_1,\ldots, V_m)&=0,\\
    P_1(V_1,\ldots, V_m)&>0,\\
    \vdots\\
    P_{\ell}(V_1,\ldots, V_m)&>0.
  \end{cases}
\end{align*}

We want to prove that we can decide whether there are distinct positive integers $x_1,\ldots, x_m$ such that
\begin{align}
  \label{eq:psi}
  \Psi(M^{x_1},\ldots, M^{x_m})
\end{align}
holds. We will simply call any such tuple $(x_1,\ldots, x_m)$ a \textbf{solution}. 

By diagonalisation, there are algebraic numbers $c_1,\ldots, c_4\in\alg$  such that for all $n\in\nat$
\begin{align*}
  M^n=\begin{pmatrix}
        c_1\lambda^n+\conj{c_1\lambda^n} & c_2\lambda^n+\conj{c_2\lambda^n}\\
        c_3\lambda^n+\conj{c_3\lambda^n} & c_4\lambda^n+\conj{c_4\lambda^n}
      \end{pmatrix}.
\end{align*}
So when polynomials $P_0,\ldots, P_{\ell}$ are instantiated with $M^x$ they can be seen as polynomials in $\lambda^x$ and $\conj\lambda^x=\lambda^{-x}$; in other words there are polynomials $Q_0,\ldots, Q_{\ell}$ with algebraic coefficients such that
\begin{align*}
  P_i(M^{x_1},\ldots, M^{x_m})=Q_i(\lambda^{x_1},\lambda^{-x_1},\ldots,\lambda^{x_m},\lambda^{-x_m}),
\end{align*}
for $0\le i\le \ell$ and all tuples of integers $(x_1,\ldots, x_m)\in\intg^m$. Let us assume at once that as part of the strict inequalities we have ones of the type
\begin{align}
  \label{eq:make sure distinct}
  \lambda^{x_j}+\lambda^{-x_j}>\lambda^{x_k}+\lambda^{-x_k},
\end{align}
for $j\ne k$ to ensure that the $x_i$ are all distinct. This is without loss of generality because any solution would certainly belong to one of these augmented semialgebraic sets. We will show how to decide if there is a solution.

\subsubsection{Proof Overview}
We begin in the next subsection by considering the case that there are only polynomial inequalities to satisfy. This case is simpler. Intuitively, on the complex plane the angles of $\lambda^n$, for integer $n$, are dense in $[0,2\pi]$, and the target set is made out of strict polynomial inequalities and therefore is open in the usual topology.  Thus, if the target set is non-empty, it contains some $\lambda^n$. The proof uses Kronecker's theorem  on simultaneous Diophantine approximation.

In subsection~\ref{sec::alg-subgroups-and-tori}, we recall the theory of algebraic subgroups and linear tori to the extent that is needed in the sequel. As we described briefly in the introduction, the reason for considering algebraic subgroups is because all $(\lambda^{x_1},\ldots, \lambda^{x_m})$ for $x_i\in\intg$ belong to an algebraic subgroup of dimension~1. We want to apply the result of Bombieri and Zannier which says that the intersection of algebraic subgroups of dimension~1 and a variety has bounded (Weil) height. For us the variety (which we denote by $X$) is the zero set of the polynomial~$Q_0$.

It is possible for a variety to contain a whole algebraic subgroup. When this happens, the height of points in the intersection may be unbounded. These cases need to be treated by separate means. This is the reason for the partition $X=X^{\circ}\cup X^{\bullet}$, as those \emph{degenerate} points are contained in $X^\bullet$. 
The end goal of subsection~\ref{sec::alg-subgroups-and-tori} is to show that we can compute the polynomial equations that define $X^\bullet$, and to state a structure theorem, giving more information about this subset.

In subsection~\ref{sec::heights} we introduce heights and the Bombieri-Zannier  theorem. In the end we tie these three subsections together by describing how the theorems can be used to decide the existence of solutions. The algorithm is conceptually simple. To check whether there is a solution in $X^\circ$ we use the height bound to derive an upper bound on the absolute value of the exponents $|x_i|$, and then simply try every one of the finitely many possibilities. If no solution is found, it remains to check whether there is one in $X^\bullet$. To this end, the algorithm constructs the defining polynomials of $X^\bullet$, and by exploiting the structure theorem, checking for solutions in $X^\bullet$ is reduced to the problem of whether there is a solution in a set defined by strict polynomial inequalities.

\subsubsection{System of Inequalities}
We prove a slightly more general result, where we allow $(x_1,\ldots,x_m)$ to range over members of a (additive) subgroup of $\intg^m$. 
\begin{lemma}
  \label{lem:no tall solutions}
  Let $\Lambda\subseteq \intg^m$ be a subgroup under addition.
  Let $\lambda\in\alg$ be as above, and suppose that we are given polynomials $S_1,\ldots,S_k$ in $2m$ variables and algebraic coefficients, such that $S_i(z_1,\conj{z_1},\ldots,z_m,\conj{z_m})$ is real-valued for all complex $z_j$ and all $i$. Then there is a procedure to decide whether there exists $(x_1,\ldots,x_m)$ in $\Lambda$, with positive coordinates, simultaneously satisfying
  \begin{align}
    \label{eq:ineqs s}
    S_i(\lambda^{x_1},\ldots, \lambda^{-x_m})>0, \text{ for all $i\in\set{1,\ldots,k}$}. 
  \end{align}
\end{lemma}
\begin{proof}
  Suppose that the subgroup $\Lambda$ is given as the integer points in the kernel of a matrix $A$ with integer entries, $m$ rows, and $m'\le m$ columns. We have:
  \begin{align*}
    \Lambda = \set{\vec x \in \intg^m\st \vec x\ A = \vec 0}.
  \end{align*}
  First check that this subgroup contains elements with positive coordinates, if it does not, clearly we answer no.
  
  Denote by $\torus$ the unit circle in the complex plane. We will write $\vec z$ for the vector $(z_1,\ldots, z_m)$ and for any vector $\vec b=(b_1,\ldots, b_m)$ of length $m$, we abbreviate
  \begin{align*}
    \vec z^{\vec b}=z_1^{b_1}\cdots z_m^{b_m}.
  \end{align*}
  Denote by $\vec a_1,\ldots,\vec a_{m'}$ the columns of $A$, and define the following semialgebraic sets:
  \begin{align*}
    \vec R &\defeq \set{\vec z \in \torus^m\st \vec z^{\vec a_i}=1\text{ for all }1\le i\le m'},\\
    \vec R' &\defeq \{\vec z \in \vec R\st S_i(z_1, z_1^{-1}, \ldots, z_m, z_m^{-1})>0,\\ &\text{ for all } 1\le i\le k\}.
  \end{align*}
  Intuitively, the set $\vec R$ is all the numbers with coordinates in the unit circle and exponents that belong to the subgroup $\Lambda$. 
  In particular, $(\lambda^{x_1},\ldots,\lambda^{x_m})$ is in $\vec R$ if and only if $\vec x\in \Lambda$. 
  Meanwhile, $\vec R'$ is the subset of such numbers that also satisfy the polynomial inequalities \eqref{eq:ineqs s}. Clearly, if $\vec R'$ is empty, there are no solutions to \eqref{eq:ineqs s}; but if it is not empty we argue below that there will always be at least one solution. Since $\vec R'$ is a semialgebraic set, we can use Tarski's algorithm to decide whether it is empty or not. 

  To show that $\vec R'\ne\emptyset$ implies the existence of a solution we use the following theorem, due to Kronecker, on simultaneous Diophantine approximations. 

  \begin{theorem}[Theorem IV in Page~53 of~\cite{cassels59_introd_to_dioph_approx}]
    Let
    \begin{align*}
      L_j(\vec y)=L_j(y_1,\ldots,y_{m'}),\ \ 1\le j\le m,
    \end{align*}
    be $m$ homogeneous linear forms in any number $m'$ of variables $y_i$. Then the two following statements about a real vector $\vec\alpha=(\alpha_1,\ldots,\alpha_m)$ are equivalent:
    \begin{enumerate}
    \item For all $\epsilon>0$ there is an integral vector $\vec a=(a_1,\ldots, a_{m'})$ such that simultaneously
      \begin{align*}
        |L_j(\vec a)-\alpha_j|<\epsilon, \ \ 1\le j\le m.
      \end{align*}
      
    \item If $\vec u=(u_1,\ldots, u_m)$ is any integral vector such that:
      \begin{align*}
        u_1L_1(\vec y)+\cdots+u_mL_m(\vec y)
      \end{align*}
      has integer coefficients, considered as a form in the indeterminates $y_i$, then
      \begin{align*}
        u_1\alpha_1+\cdots+u_m\alpha_m\in\intg.
      \end{align*}
    \end{enumerate}
  \end{theorem}

  In order to apply this theorem, we define our linear forms $L_i$ as follows. By putting $A$ in a row-reduced echelon form, finding a basis and multiplying with a suitable scalar, we can compute a set of integral vectors $b_1,\ldots, b_{m'}$ that generate $\Lambda$. Write $\lambda=\exp(\vartheta 2\pi\ii)$ and suppose that $\vartheta$ is irrational (otherwise $\lambda$ is a root of 1 and the lemma is trivial). For $1\le j\le m$ define:
  \begin{align*}
    L_j(y_1,\ldots, y_{m'})\defeq \sum_{i=1}^{m'}\vartheta\ b_{i,j}\ y_i. 
  \end{align*}
  Suppose that $\vec R'$ is nonempty, and choose some element $\zeta\in\vec R'$ and write it as:
  \begin{align*}
\zeta=    \big(\exp(\alpha_1 2\pi\ii), \ldots, \exp(\alpha_m 2\pi\ii)\big).
  \end{align*}
  Let $\vec u = (u_1,\ldots, u_m)\in\intg^m$ be such that $\sum u_iL_i(\vec y)$ has integer coefficients, considered as a form in the indeterminates~$y_i$. A short computation shows that since $\vartheta$ is irrational, for such $\vec u$ we have
  \begin{align*}
    \vec u\ B = \vec 0,
  \end{align*}
  where $B$ is the matrix that has the vectors $b_1,\ldots, b_{m'}$ as columns. This means that such vectors $\vec u$ belong to the orthogonal complement of the linear subspace $V\subseteq \rel^{m}$, spanned by $b_1, \ldots, b_{m'}$. By virtue of $\zeta$ belonging to $\vec R'$ and hence also $\vec R$, we have that $(\alpha_1,\ldots, \alpha_m)$ belongs to $V$, and consequently $\sum u_i\alpha_i=0$. We have proved that Statement~2 in the above theorem holds for our real vector $\alpha$. Applying the theorem gives us Statement~1, namely that there are integral vectors $\vec a$ that make $L_j(\vec a)$ get arbitrarily close to $\alpha_j$. As $\vec a$ ranges over $\intg^{m'}$, $(L_1(\vec a), \ldots, L_m(\vec a))$ ranges over $\vartheta\Lambda$, which in turn means that
  \begin{align}
    \label{eq:in R}
  (\lambda^{L_1(\vec a)/\vartheta}, \ldots, \lambda^{L_m(\vec a)/\vartheta})\in\vec R,
\end{align}
and gets arbitrarily close to $\zeta$. Finally, since $\vec R'$ is an open subset of $\vec R$, by choosing $\epsilon$ small enough, we get some $\vec a$ such that the tuple of~\eqref{eq:in R} belongs to the subset $\vec R'$.

It remains to check that we  can find one such $\vec a$ such that the exponents in \eqref{eq:in R} are positive. We know that the subgroup $\Lambda$ has elements with positive coordinates, and this is in fact sufficient, due to an equidistribution theorem of Weyl that can be found in the section starting at Page~64 of~\cite{cassels59_introd_to_dioph_approx}. 
\end{proof}
\subsubsection{Algebraic Subgroups and Tori}
\label{sec::alg-subgroups-and-tori}
We begin with a few definitions. The general theory is developed more extensively in \cite{schmidt}, \cite{schinzel2000polynomials}, and especially in \cite[Chapter~3]{bombieri2007heights}. We borrow from the latter freely. 

It is convenient in the rest of this section to set $n:=2m$, where $m$ is the number of times we want to enter the target set.  A \textbf{variety} $Y$ in affine $n$-dimensional space $\alg^n$ is defined to be the set of solutions $(y_1,\ldots, y_n)$ of a system of polynomial equations $f_i(y_1,\ldots, y_n)=0$, where each $f_i$ has algebraic coefficients. We say that a variety is \textbf{irreducible} if it cannot be written as the union of two proper subvarieties.

We define $\multg n$ to be the set of tuples $(z_1,\ldots, z_n)$ of non-zero algebraic numbers.
It has a group structure under component-wise multiplication:
\begin{align*}
  (y_1,\ldots, y_n) \cdot (z_1,\ldots, z_n) = (y_1z_1,\ldots, y_nz_n). 
\end{align*}

The variety that we are interested in, which we will denote by $X\subseteq \multg n$, is the zero set of our polynomial $Q_0$, conjoined with polynomial equations
\begin{align*}
  z_jz_{j+1}-1=0,
\end{align*}
where $1\le j\le n$ is an odd number, to ensure that the conjugacy relations hold. We assume that $X$ is irreducible, for otherwise we can factorize the polynomials and treat the irreducible components in turn. We will effectively find all points in the intersection of this variety and some algebraic subgroup of dimension 1, which we now define. 

An \textbf{algebraic subgroup} is a subvariety of $\multg n$ that is also a subgroup. As an example, given an additive subgroup $\Lambda\subseteq \intg^n$, we can see that it determines an algebraic subgroup
\begin{align*}
  H_{\Lambda}\defeq \set{(z_1,\ldots, z_n)\in \multg n : z_1^{a_1}z_2^{a_2}\cdots z_n^{a_n}=1\text{ for all }\vec a \in \Lambda}. 
\end{align*}
In fact every algebraic subgroup is of this type, \cite[Corollary 3.2.15]{bombieri2007heights}. Further, if $\Lambda$ is a subgroup of $\intg^n$ of rank $n-r$ then $H_\Lambda$ is an algebraic subgroup of dimension $r$. By dimension here we mean the dimension of the variety, see for example~\cite[Page~5]{hartshorne77_algeb_geomet}. One way of defining the dimension of a variety $X$ is as the maximum length of a  chain $X_0\subset X_1\subset \cdots \subset X_k$ of irreducible subvarieties of $X$.

We prove that powers of $\lambda$ belong to algebraic subgroups of dimension~1, as remarked in the introduction. 
\begin{lemma}
  \label{lem:belong to algebraic groups of dim 1}
  For all $(a_1,\ldots, a_k)\in\intg^k$, the point
  \begin{align*}
    \left(\lambda^{a_1},\ldots, \lambda^{a_k}\right)
  \end{align*}
  belongs to an algebraic subgroup of dimension 1. 
\end{lemma}
\begin{proof}
  If all $a_i=0$, then the lemma clearly holds, so suppose that there is some $j$ such that $a_j\ne 0$.
  The subgroup $\Lambda =\{(b_1,\ldots,b_k)\in \mathbb Z^k : a_1b_1+\cdots+a_kb_k\}$
  has rank $k-1$ and hence the group $H_\Lambda$, which contains the point in the statement of the lemma, has dimension~1. See \cite[Proposition~3.2.7]{bombieri2007heights}. 
  %
\end{proof}
We denote by $\mathcal H_1(n)$ the union of all algebraic subgroups of $\multg n$ that have dimension~1; the parameter $n$ will be omitted when the ambient dimension is understood. We are interested in the intersection
\begin{align*}
  \mathcal H_1 \cap X, 
\end{align*}
as, according to the lemma above, this set contains all
\begin{align*}
(\lambda^{x_1},\lambda^{-x_1},\ldots,\lambda^{x_m},\lambda^{-x_m})  
\end{align*}
for which
\begin{align*}
 Q_0(\lambda^{x_1},\lambda^{-x_1},\ldots,\lambda^{x_m},\lambda^{-x_m})=0, 
\end{align*}
where $x_i$ are integers.

In order to analyse the intersection above, the variety $X$ will be partitioned into two subsets which we now define. A \textbf{linear torus} is an algebraic subgroup that is irreducible. A \textbf{torus coset} is a coset of the form $g H$ where $H \subseteq \multg n$ is a linear torus and $g \in \multg n$.

We denote by $X^\bullet$ the union of all nontrivial torus cosets that are contained entirely in $X$, in other words:
\begin{align*}
  X^\bullet\defeq \bigcup\set{gH \text{ a non-trivial torus coset }\st gH\subseteq X }.
\end{align*}
Also define
\begin{align*}
  X^{\circ} \defeq X\setminus X^\bullet. 
\end{align*}
Below we give an effective characterisation of $X^\bullet$.

Recall that for a vector of integers $\vec a\in\intg^n$ we write
\begin{align*}
  \vec z^{\vec a} = z_1^{a_1}\cdots z_n^{a_n}. 
\end{align*}
Let $A$ be an $n \times n$ matrix with integer entries, and denote by $A_1,\ldots, A_n$ its columns. We write by $\varphi_A\st \multg n \to \multg n$ the map
\begin{align*}
  \varphi_A(\vec z) \defeq \big(\vec z^{A_1}, \ldots, \vec z^{A_n} \big).
\end{align*}
One can show that $\varphi_{AB}=\varphi_{B}\circ \varphi_{A}$, and as a consequence for matrices $A$ with determinant $\pm 1$, $\varphi_{A}$ is an isomorphism\footnote{This means that it is a group homomorphism that is also a morphism of algebraic varieties.} with inverse $\varphi_{A^{-1}}$. Such an isomorphism is called a \textbf{monoidal transformation}.
The group of $n\times n$ integer matrices with determinant $\pm 1$ is the special linear group, denoted $\mathrm{SL}(n, \intg)$.

We state here some important basic results related to the structure of algebraic subgroups. Recall that we have used the notation $\modl{\vec a}$ for the $\ell^1$ norm; when $A$ is a matrix, we denote by $\modl A$ the maximum of $\ell^1$ norms of its columns.

\begin{proposition}[{\cite[Proposition 3.2.10 and Corollary 3.2.9]{bombieri2007heights}}]
  \label{prop:torus isomorphism}
  Let $H_\Lambda$ be a linear torus, where $\Lambda$ is a subgroup of\, $\intg^n$ of rank $n-r$ and suppose that $\Lambda$ has $n-r$ independent vectors of norm at most~$N$. Then there is a matrix $A\in\mathrm{SL}(n,\intg)$ with $\modl A\le n^3N^{n-r}$ and $\modl {A^{-1}}\le n^{2n-1}N^{(n-1)^2}$, such that
  \begin{align*}
    \varphi_A(\vec 1_{n-r}\times \multg r) = H_{\Lambda},
  \end{align*}
  where
   $ \vec 1_{n-r} \defeq \{(1,\ldots,1)\}$.
\end{proposition}
From the bounds on $A$, we can effectively compute such a matrix  given $n-r$ independent vectors of $\Lambda$.
Next, let $X\subseteq \multg n$  be our subvariety. We say that an algebraic subgroup $H$ of $\multg n$ is \textbf{maximal} in $X$ if $H\subseteq X$ and $H$ is not contained in a larger subgroup that is contained in $X$.

\begin{proposition}[{\cite[Proposition 3.2.14]{bombieri2007heights}}]
  \label{prop:maximal subgroups}
  Let $X\subseteq \multg n$ be a subvariety, defined by polynomial equations $f_i(\vec x):=\sum c_{i,\vec a}\vec x^{\vec a}=0$, $1\le i\le k$, and let $E_i$ be the set of exponents appearing in the monomials of $f_i$. Let $H$ be a maximal algebraic subgroup of $\multg n$ contained in $X$. Then $H=H_{\Lambda}$ where $\Lambda$ is generated by vectors of type $\vec a'_i-\vec a_i$, with $\vec a'_i, \vec a_i\in E_i$, for $i=1,\ldots,k$. 
\end{proposition}

The first proposition above says that linear tori of dimension $r$ are simply isomorphic to $\multg r$, and that the isomorphism is given in terms of a monoidal transformation that we can compute.
The second proposition tells us that maximal algebraic subgroups contained in a variety $X$ are defined simply by the exponents of monomials that appear in the definition of $X$.

The two propositions above have the following important consequence. If $gH\subseteq X$ is a maximal torus coset (meaning that it is not contained in another torus coset), then $H$ is one of the components of a maximal algebraic subgroup $H'$ of the variety $g^{-1}X$. \Cref{prop:maximal subgroups} implies that there are finitely many such $H'$, that we can effectively compute them, and further that they are independent of $g$---note that only the exponents matter in the proposition, not the coefficients. Since it is possible to compute the equations of each component of $H'$ by factoring in the number field $\rat(\lambda)$, we have: 
\begin{lemma}
  \label{lem:maximal torus cosets}
  We can effectively construct a (possibly empty) set $\mathcal T_X$ of positive-dimensional linear tori such that if $gH\subseteq X$ is a maximal torus coset, then $H\in\mathcal T_X$, and for every $H\in\mathcal T_X$ there is some torus coset $gH\subseteq X$ which is maximal.
\end{lemma}
From this lemma, another way of defining the subset $X^{\bullet}$ is
\begin{align*}
  X^{\bullet}=\bigcup\set{gH\st g\in\multg n, H\in\mathcal T_X,\text{ and }gH\subseteq X}.
\end{align*}
Although we can effectively construct the subgroups $H$, we do not yet have an effective method of constructing the union of all maximal cosets $gH$ that are contained in $X$. This is done in the following lemma. 
\begin{lemma}[{\cite[Theorem 3.3.9]{bombieri2007heights}}]
  \label{lem:torus coset subvariety}
  Let $X\subseteq \multg n$ be a subvariety and $H$ a linear torus of dimension $r\ge 1$. Then there exists a matrix $A\in\mathrm{SL}(n,\intg)$, which can be computed, such that
  \begin{align*}
    \bigcup_{gH\subseteq X}gH = \varphi_A(X_1\times \multg r),
  \end{align*}
  where $X_1\subseteq \multg{n-r}$ is a subvariety, whose  defining polynomials can be computed.
\end{lemma}
\begin{proof}
  Using \Cref{prop:maximal subgroups} we can conclude that $H=H_{\Lambda}$ where $\Lambda$ is a subgroup of $\intg^n$ of rank $n-r$, and from \Cref{prop:torus isomorphism}, we can compute a matrix $A$, such that $H=\varphi_A(\vec 1_{n-r}\times \multg r)$. If we define $\widetilde X$ to be $\varphi_{A}^{-1}(X)$, we have
  \begin{align*}
    \bigcup_{gH\subseteq X}gH = \bigcup_{g\cdot (\vec 1_{n-r}\times \multg r)\subseteq \widetilde X}g\cdot (\vec 1_{n-r}\times \multg r). 
  \end{align*}
  Note that since $A$ can be computed, so can the defining polynomials of~$\widetilde X$. Let $f_1,\ldots, f_k$ be these defining polynomials of $\widetilde X$. Then $g\cdot (\vec 1_{n-r}\times \multg r)$ being a subset of $\widetilde X$ means that
\begin{align*}
    f_i(g_1,\ldots, g_{n-r}, y_{n-r+1},\ldots, y_n)=0,\qquad 1\le i\le k,
\end{align*}
are identically satisfied in $y_{n-r+1},\ldots, y_n$. This is just a set of polynomial equations in indeterminates $g_1,\ldots, g_{n-r}$, \ie a subvariety of $\multg {n-r}$, which we call $X_1$. So if $g\in X_1$, then $g\cdot (\vec 1_{n-r}\times \multg r)\subseteq \widetilde X$, or equivalently $\varphi_A(g\cdot (\vec 1_{n-r}\times \multg r))\subseteq X$. The lemma follows. 
\end{proof}

To summarise, in this section we proved that (i) we can compute the finite set of subgroups $H$, such that $gH$ is a maximal coset contained in $X$. We called this finite set of subgroups $\mathcal T_X$. We also showed (ii) that for any $H\in\mathcal T_X$ the union of all maximal cosets $gH$ that are contained in $X$ are isomorphic to $X_1\times \multg r$ for some $r\ge 1$. Furthermore, the defining equations of $X_1$ can also be computed and the isomorphism map too.

These facts give sufficient information to decide if there are any solutions in $X^\bullet$. Next, we discuss heights and the Bombieri-Zannier theorem. 

\subsubsection{Heights}
\label{sec::heights}
The height of a point $\vec z$ in $\alg^n$ is a central notion in Diophantine geometry. It is used to measure the arithmetic complexity of $\vec z$. For more details the reader should consult, for example, Chapter~1 of \cite{bombieri2007heights}. For our purposes, it suffices to define the height as follows. Let $K:=\rat(\lambda)$ be the number field that we work in.
Let $M_K$ be a set of absolute values satisfying the product formula.  Define
\begin{align*}
  \log^+t\defeq\max(0,\log t).
\end{align*}
Then the \textbf{height}\footnote{The full name is the absolute logarithmic Weil height.} of a point $\vec z=(z_1,\ldots, z_n)\in K^n$ is defined as:
\begin{align*}
  h(\vec z)\defeq \sum_{v\in M_K}\max_j\log^+|z_j|_v. 
\end{align*}
We are interested in specific points of the form $(\lambda^{x_1},\ldots,\lambda^{x_n})$, where $x_i\in\intg$. The height of such points has the following properties:
\begin{lemma}
  \label{lem:bound}
  Let $\vec x\in\intg^n$, and denote by $M = \max_j |x_j|$. Then
  \begin{align*}
    M h(\lambda) \le h\big((\lambda^{x_1},\ldots, \lambda^{x_n})\big)\le 2Mh(\lambda). 
  \end{align*}
\end{lemma}
\begin{proof}
  By the definition of height and absolute value we have:
  \begin{align*}
    h\big((\lambda^{x_1},\ldots, \lambda^{x_n})\big)&=\sum_{v\in M_K}\max_j\log^+|\lambda^{x_j}|_v\\
    &=\sum_{v\in M_K}\max_j\log^+|\lambda|_v^{x_j}.
  \end{align*}
  Since for every absolute value $|\cdot|_v$, $|\lambda|_v|\lambda^{-1}|_v=1$, it follows that
  \begin{align*}
    \sum_{v\in M_K}\max_j\log^+|\lambda|_v^{x_j}\le M(h(\lambda)+h(\lambda^{-1})),
  \end{align*}
  and since $h(\alpha)=h(\alpha^{-1})$ for every algebraic $\alpha$ (see \cite[Lemma~1.5.18]{bombieri2007heights}), we get the upper bound. For the lower bound we have:
  \begin{align*}
    h\big((\lambda^{x_1},\ldots, \lambda^{x_n})\big) \ge h(\lambda^M) = Mh(\lambda).
  \end{align*}
\end{proof}

The main fact that allows for a procedure to decide multiple reachability for rotations is the following theorem on heights of points in $X^\circ \cap \mathcal H_1$, due to Bombieri and Zannier:
\begin{theorem}[{\cite[Theorem~1, Page~524]{schinzel2000polynomials}}]
  \label{thm:height upper bound}
  Let $X\subseteq \multg n$ be a subvariety. Then there exists an effective bound $b\in\nat$ depending only on $X$ such that for all algebraic points
$\vec z in  X^{\circ}\cap\mathcal H_1$ we have $h(\vec z)\le b. $.
\end{theorem}
The theorem cited in \cite{schinzel2000polynomials} does not explicitly state that the bound is effective, but upon a closer inspection of the proof one can see that all the bounds are explicit, with the exception of the points $(c_1^*,\ldots,c_h^*)\in\intg^h$ which are chosen to be outside a finite number of linear subspaces of $\rat^h$. It is plain that we can effectively construct such a point. 

Now we have all the tools to describe the algorithm and justify its correctness. 
\subsubsection{Algorithm}
\label{sec::alg}
The procedure first searches for solutions in $X^\circ$. Let $b\in\nat$ be an upper bound on the height of algebraic points in the intersection of $\mathcal H_1$ (which is the union of all subgroups of dimension~1) and $X^\circ$. Such an upper bound can be computed with \Cref{thm:height upper bound}. From \Cref{lem:belong to algebraic groups of dim 1} we know that for all integers $x_1,\ldots, x_m\in\intg$, the algebraic points
\begin{align}
  \label{eq:lambda power}
  \big(\lambda^{x_1},\lambda^{-x_1},\ldots, \lambda^{x_m},\lambda^{-x_m}\big)
\end{align}
all belong to $\mathcal H_1$. So if any of the points in \eqref{eq:lambda power} is in $X^\circ$ it is also in the intersection $X^\circ\cap \mathcal H_1$. The upper bound $b$ on the height of points in this intersection translates to an upper bound on the exponents $\modl {\vec x}$ due to \Cref{lem:bound}. It remains to check whether any of the finitely many points \eqref{eq:lambda power}, with $\modl{\vec x}\le b$, satisfy the polynomial equality $Q_0=0$ and inequalities $Q_i>0$, $1\le i\le \ell$. These checks are performed by using Tarski's algorithm, making sure that the exponents are distinct and positive. If a solution is found, we return yes, otherwise we continue the search in $X^\bullet$.

Now, by applying \Cref{lem:maximal torus cosets} we compute the defining polynomials of tori that are in the set $\mathcal T_{X}$. Recall that this set contains all positive dimensional tori, which have a maximal coset entirely contained in $X$. If $\mathcal T_{X}$ is empty, so is the set $X^{\bullet}$, and we are done: the algorithm returns no, because no solutions were found in $X^\circ$ and $X^\bullet=\emptyset$.

So suppose that $\mathcal T_X$ is nonempty. The procedure searches for solutions in all the elements of $\mathcal T_X$ in turn, in the following way. Let $H\in\mathcal T_X$ be an element, of dimension $r$. By definition, this means that $H$ is a torus for which there is a maximal coset $gH$ entirely contained in $X$, and $r\ge 1$.

If $H$ has dimension $r=n$, then this essentially means that $X^\bullet=\multg n$ and hence $X=X^\bullet$, which in turn implies that the polynomial $Q_0$ is identically~0. So it is only the strict polynomial inequalities $Q_i>0$ that need to hold for there to be a solution. We can check whether the inequalities can be satisfied by applying \Cref{lem:no tall solutions}, with $\Lambda=\intg^m$, and polynomials $Q_1,\ldots, Q_\ell$. Recall that the requirement for the exponents to be distinct is assumed to be encoded in the polynomial inequalities, as remarked in the beginning of this section. This concludes the case when $H$ has dimension $r=n$.

We assume now that $H$ has dimension $r$, where $0<r<n$. Using \Cref{lem:torus coset subvariety}, we next compute a matrix $A\in\mathrm{SL}(n,\intg)$, and a subvariety $X_1\subseteq \multg {n-r}$, such that
\begin{align*}
  \bigcup_{gH\subseteq X}gH=\varphi_A(X_1\times \multg r). 
\end{align*}

Now $X_1\subseteq\multg {n-r}$ does not contain any positive dimensional coset, \ie $X_1=X_1^\circ$. To see this, assume towards a contradiction that there is some $g_1$ and a torus $H_1$ of dimension $r_1>0$ such that $g_1H_1\subseteq X_1$. Then we have
\begin{align*}
  \bigcup_{gH\subseteq X}gH\supseteq \varphi_A(g_1H_1\times \multg r).
\end{align*}
\Cref{prop:torus isomorphism} implies that there exists a monoidal transformation $\varphi_1$ such that
\begin{align}
  \label{eq:inclusion}
  \bigcup_{gH\subseteq X}gH\supseteq \varphi_A\big(\varphi_1(\vec 1_{n-r-r_1}\times \multg {r_1})\times \multg r\big).
\end{align}
From the proof of \Cref{lem:torus coset subvariety} it plainly follows that there is a bijection between points in $X_1$ and cosets $gH$ that are contained in $X$. This fact together with the inclusion \eqref{eq:inclusion} yield the existence of a coset $gH$ of dimension~$r$ that is contained in a coset of dimension $r+r_1$, both of which are inside $X$. Since $r_1>0$, the coset $gH$ is contained in a strictly larger coset; contradicting the definition of $\mathcal T_X$ which says that all $gH$ should be maximal. 

Now we can write the union of all cosets $gH$ contained in~$X$~as
\begin{align*}
\varphi_A(X_1^\circ\times \multg r).  
\end{align*}

The union of all subgroups of dimension one, $\mathcal H_1$, is invariant under monoidal transformations.
Therefore,
\begin{align*}
  \mathcal H_1\cap \varphi_A(X_1^\circ\times \multg r) = \varphi_A(\mathcal H_1)\cap \varphi_A(X_1^\circ\times \multg r),
\end{align*}
and since $\varphi_A$ is an isomorphism we have
\begin{align*}
  = \varphi_A\big(\mathcal H_1\cap (X_1^{\circ}\times \multg r)\big). 
\end{align*}
Through composition with the polynomial defining the monoidal transformation $\varphi_A$, we can construct polynomials $\widetilde{Q_0},\ldots,\widetilde{Q_\ell}$ such that if
\begin{align*}
  \vec z \in \varphi_A\big(\mathcal H_1\cap (X_1^{\circ}\times \multg r)\big)
\end{align*}
satisfies the polynomial (in)equalities $Q_0=0$, $Q_i>0$ for $1\le i\le \ell$, then
\begin{align*}
  \varphi_{A^{-1}}(\vec z) \in \mathcal H_1\cap (X_1^{\circ}\times \multg r), 
\end{align*}
satisfies the polynomial (in)equalities $\widetilde{Q_0}=0$, $\widetilde{Q_i}>0$, for $1\le i\le \ell$. Using the procedure of \Cref{thm:height upper bound} we compute a bound $b_1\in\nat$ for the intersection
\begin{align*}
  \mathcal H_1(n-r)\cap X_1^\circ
\end{align*}
where $ \mathcal H_1(n-r)$ is the union of all algebraic subgroups of $\multg{n-r}$.
As above, we search for $(\lambda^{x_1},\ldots, \lambda^{x_{n-r}})$ with $\modl x \le b_1$ that belong to $X_1^\circ$. If none are found, the procedure halts and returns no. Since $\varphi_A$ sends powers of $\lambda$ to powers of $\lambda$, the no answer is justified, as indeed there are no solutions.

If a finite number of $(\lambda^{x_1},\ldots, \lambda^{x_{n-r}})$ belonging to $X_1^\circ$ are found, we try each in turn to see if they can be made to satisfy the inequalities as well. Let $(\lambda^{x_1},\ldots, \lambda^{x_{n-r}})$ be one such point. Fixing the first $n-r$ coordinates to these powers~of $\lambda$ in the polynomial (in)equalities makes $\widetilde{Q_0}$ identically zero, and gives us new inequalities $R_i>0$, $1\le i\le \ell$. By construction, the polynomials $R_i$ will satisfy the hypothesis of \Cref{lem:no tall solutions}, so we can apply this lemma for $\Lambda=\intg^{r}$ to determine if $R_1,\ldots,R_\ell$ are satisfied by some powers of $\lambda$. 
If such powers of $\lambda$ are found, the procedure halts and returns yes\footnote{A detail that needs to be justified is that $x_i$ need to be positive after the application of $\varphi_{A^{-1}}$. But we can find such $x_i$ due to the equidistribution theorem that was used in the proof of \Cref{lem:no tall solutions}.}. 
If there is not, we continue with another candidate $(\lambda^{y_1},\ldots, \lambda^{y_{n-r}})$ that has $\modl x \le b_1$, and which belongs to $X_1^\circ$. 
This concludes the proof of \Cref{thm:pos1}(ii). 

We briefly comment about why we are limited to rotations on the plane. If the given matrix is not a rotation, then the relevant points do not all belong to $\mathcal H_1$, but rather to $\mathcal H_2$, in subgroups of dimension 2. Intuitively this is because the matrix changes vectors over two dimensions: scaling and rotating.
While there are finiteness results, often as special cases of the Mordell-Lang conjecture, see \eg \cite{laurent1984equations}, but to our knowledge
for subgroups of dimension 2 we lack an effective height bound akin to that in \Cref{thm:height upper bound}.



\bibliographystyle{ACM-Reference-Format}
\bibliography{bibliography}


\begin{thebibliography}{21}


\ifx \showCODEN    \undefined \def \showCODEN     #1{\unskip}     \fi
\ifx \showDOI      \undefined \def \showDOI       #1{#1}\fi
\ifx \showISBNx    \undefined \def \showISBNx     #1{\unskip}     \fi
\ifx \showISBNxiii \undefined \def \showISBNxiii  #1{\unskip}     \fi
\ifx \showISSN     \undefined \def \showISSN      #1{\unskip}     \fi
\ifx \showLCCN     \undefined \def \showLCCN      #1{\unskip}     \fi
\ifx \shownote     \undefined \def \shownote      #1{#1}          \fi
\ifx \showarticletitle \undefined \def \showarticletitle #1{#1}   \fi
\ifx \showURL      \undefined \def \showURL       {\relax}        \fi
\providecommand\bibfield[2]{#2}
\providecommand\bibinfo[2]{#2}
\providecommand\natexlab[1]{#1}
\providecommand\showeprint[2][]{arXiv:#2}

\bibitem[Almagor et~al\mbox{.}(2019)]%
        {shaullstacs}
\bibfield{author}{\bibinfo{person}{Shaull Almagor}, \bibinfo{person}{Jo{\"{e}}l
  Ouaknine}, {and} \bibinfo{person}{James Worrell}.}
  \bibinfo{year}{2019}\natexlab{}.
\newblock \showarticletitle{The Semialgebraic Orbit Problem}. In
  \bibinfo{booktitle}{\emph{36th International Symposium on Theoretical Aspects
  of Computer Science, {STACS} 2019, March 13-16, 2019, Berlin, Germany}}
  \emph{(\bibinfo{series}{LIPIcs}, Vol.~\bibinfo{volume}{126})},
  \bibfield{editor}{\bibinfo{person}{Rolf Niedermeier} {and}
  \bibinfo{person}{Christophe Paul}} (Eds.). \bibinfo{publisher}{Schloss
  Dagstuhl - Leibniz-Zentrum f{\"{u}}r Informatik}, \bibinfo{pages}{6:1--6:15}.
\newblock
\urldef\tempurl%
\url{https://doi.org/10.4230/LIPIcs.STACS.2019.6}
\showDOI{\tempurl}


\bibitem[Baker(1990)]%
        {baker1990transcendental}
\bibfield{author}{\bibinfo{person}{Alan Baker}.}
  \bibinfo{year}{1990}\natexlab{}.
\newblock \bibinfo{booktitle}{\emph{Transcendental number theory}}.
\newblock \bibinfo{publisher}{Cambridge university press}.
\newblock


\bibitem[Bombieri and Gubler(2007)]%
        {bombieri2007heights}
\bibfield{author}{\bibinfo{person}{Enrico Bombieri} {and}
  \bibinfo{person}{Walter Gubler}.} \bibinfo{year}{2007}\natexlab{}.
\newblock \bibinfo{booktitle}{\emph{Heights in Diophantine geometry}}.
\newblock \bibinfo{publisher}{Cambridge university press}.
\newblock


\bibitem[Brindza et~al\mbox{.}(2001)]%
        {brindza2001multiplicities}
\bibfield{author}{\bibinfo{person}{B. Brindza}, \bibinfo{person}{{\'A}.
  Pint{\'e}r}, {and} \bibinfo{person}{W.~M. Schmidt}.}
  \bibinfo{year}{2001}\natexlab{}.
\newblock \showarticletitle{Multiplicities of binary recurrences}.
\newblock \bibinfo{journal}{\emph{Canad. Math. Bull.}} \bibinfo{volume}{44},
  \bibinfo{number}{1} (\bibinfo{year}{2001}), \bibinfo{pages}{19--21}.
\newblock


\bibitem[Cassels(1959)]%
        {cassels59_introd_to_dioph_approx}
\bibfield{author}{\bibinfo{person}{J.~W.~S. Cassels}.}
  \bibinfo{year}{1959}\natexlab{}.
\newblock \bibinfo{booktitle}{\emph{An Introduction To Diophantine
  Approximation}}.
\newblock


\bibitem[Cerlienco et~al\mbox{.}(1987)]%
        {cerlienco1987linear}
\bibfield{author}{\bibinfo{person}{L. Cerlienco}, \bibinfo{person}{M.
  Mignotte}, {and} \bibinfo{person}{F. Piras}.}
  \bibinfo{year}{1987}\natexlab{}.
\newblock \showarticletitle{Linear recurrent sequences: algebraic and
  arithmetical properties}.
\newblock \bibinfo{journal}{\emph{Enseign. Math.(2)}} \bibinfo{volume}{33},
  \bibinfo{number}{1-2} (\bibinfo{year}{1987}), \bibinfo{pages}{67--108}.
\newblock


\bibitem[Davis et~al\mbox{.}(1976)]%
        {matijasevic}
\bibfield{author}{\bibinfo{person}{Martin Davis}, \bibinfo{person}{Yuri
  Matijasevic}, {and} \bibinfo{person}{Julia Robinson}.}
  \bibinfo{year}{1976}\natexlab{}.
\newblock \bibinfo{booktitle}{\emph{Hilbert’s tenth problem: Diophantine
  equations: positive aspects of a negative solution}}.
\newblock 323--378 pages.
\newblock
\urldef\tempurl%
\url{https://doi.org/10.1090/pspum/028.2/0432534}
\showDOI{\tempurl}


\bibitem[Derksen and Masser(2015)]%
        {derksenmasser15}
\bibfield{author}{\bibinfo{person}{H. Derksen} {and} \bibinfo{person}{D.
  Masser}.} \bibinfo{year}{2015}\natexlab{}.
\newblock \showarticletitle{Linear equations over multiplicative groups,
  recurrences, and mixing II}.
\newblock \bibinfo{journal}{\emph{Indagationes Mathematicae}}
  \bibinfo{volume}{26}, \bibinfo{number}{1} (\bibinfo{date}{Jan}
  \bibinfo{year}{2015}), \bibinfo{pages}{113–136}.
\newblock
\showISSN{0019-3577}
\urldef\tempurl%
\url{https://doi.org/10.1016/j.indag.2014.08.002}
\showDOI{\tempurl}


\bibitem[Everest et~al\mbox{.}(2003)]%
        {recseq}
\bibfield{author}{\bibinfo{person}{Graham Everest}, \bibinfo{person}{Alf
  van~der Poorten}, \bibinfo{person}{Igor Shparlinski}, {and}
  \bibinfo{person}{Thomas Ward}.} \bibinfo{year}{2003}\natexlab{}.
\newblock \bibinfo{booktitle}{\emph{Recurrence Sequences}}.
\newblock
\urldef\tempurl%
\url{https://doi.org/10.1090/surv/104}
\showDOI{\tempurl}


\bibitem[Gimbert and Oualhadj(2010)]%
        {gimbert-PA}
\bibfield{author}{\bibinfo{person}{Hugo Gimbert} {and}
  \bibinfo{person}{Youssouf Oualhadj}.} \bibinfo{year}{2010}\natexlab{}.
\newblock \showarticletitle{Probabilistic Automata on Finite Words: Decidable
  and Undecidable Problems}. In \bibinfo{booktitle}{\emph{Automata, Languages
  and Programming}}, \bibfield{editor}{\bibinfo{person}{Samson Abramsky},
  \bibinfo{person}{Cyril Gavoille}, \bibinfo{person}{Claude Kirchner},
  \bibinfo{person}{Friedhelm Meyer auf~der Heide}, {and}
  \bibinfo{person}{Paul~G. Spirakis}} (Eds.). \bibinfo{publisher}{Springer
  Berlin Heidelberg}, \bibinfo{address}{Berlin, Heidelberg},
  \bibinfo{pages}{527--538}.
\newblock
\showISBNx{978-3-642-14162-1}


\bibitem[Hartshorne(1977)]%
        {hartshorne77_algeb_geomet}
\bibfield{author}{\bibinfo{person}{Robin Hartshorne}.}
  \bibinfo{year}{1977}\natexlab{}.
\newblock \bibinfo{booktitle}{\emph{Algebraic Geometry}}.
\newblock
\urldef\tempurl%
\url{https://doi.org/10.1007/978-1-4757-3849-0}
\showDOI{\tempurl}


\bibitem[Jones(1982)]%
        {jones}
\bibfield{author}{\bibinfo{person}{James~P. Jones}.}
  \bibinfo{year}{1982}\natexlab{}.
\newblock \showarticletitle{Universal Diophantine Equation}.
\newblock \bibinfo{journal}{\emph{The Journal of Symbolic Logic}}
  \bibinfo{volume}{47}, \bibinfo{number}{3} (\bibinfo{year}{1982}),
  \bibinfo{pages}{549--571}.
\newblock
\showISSN{00224812}
\urldef\tempurl%
\url{http://www.jstor.org/stable/2273588}
\showURL{%
\tempurl}


\bibitem[Kannan and Lipton(1986)]%
        {kannan86_polyn_time_algor_orbit_probl}
\bibfield{author}{\bibinfo{person}{R. Kannan} {and} \bibinfo{person}{R.~J.
  Lipton}.} \bibinfo{year}{1986}\natexlab{}.
\newblock \showarticletitle{Polynomial-Time Algorithm for the Orbit Problem}.
\newblock \bibinfo{journal}{\emph{J. ACM}} \bibinfo{volume}{33},
  \bibinfo{number}{4} (\bibinfo{year}{1986}), \bibinfo{pages}{808--821}.
\newblock
\urldef\tempurl%
\url{https://doi.org/10.1145/6490.6496}
\showDOI{\tempurl}


\bibitem[Karimov et~al\mbox{.}(2022)]%
        {Karimov2022}
\bibfield{author}{\bibinfo{person}{Toghrul Karimov}, \bibinfo{person}{Edon
  Kelmendi}, \bibinfo{person}{Jo{\"e}l Ouaknine}, {and} \bibinfo{person}{James
  Worrell}.} \bibinfo{year}{2022}\natexlab{}.
\newblock \bibinfo{booktitle}{\emph{What's Decidable About Discrete Linear
  Dynamical Systems?}}
\newblock \bibinfo{publisher}{Springer Nature Switzerland},
  \bibinfo{address}{Cham}, \bibinfo{pages}{21--38}.
\newblock
\showISBNx{978-3-031-22337-2}
\urldef\tempurl%
\url{https://doi.org/10.1007/978-3-031-22337-2_2}
\showDOI{\tempurl}


\bibitem[Laurent(1984)]%
        {laurent1984equations}
\bibfield{author}{\bibinfo{person}{Michel Laurent}.}
  \bibinfo{year}{1984}\natexlab{}.
\newblock \showarticletitle{Equations diophantiennes exponentielles}.
\newblock \bibinfo{journal}{\emph{Inventiones mathematicae}}
  \bibinfo{volume}{78} (\bibinfo{year}{1984}), \bibinfo{pages}{299--327}.
\newblock


\bibitem[Lockwood(1967)]%
        {lockwood1967book}
\bibfield{author}{\bibinfo{person}{Edward~Harrington Lockwood}.}
  \bibinfo{year}{1967}\natexlab{}.
\newblock \bibinfo{booktitle}{\emph{A book of curves}}.
\newblock \bibinfo{publisher}{Cambridge University Press}.
\newblock


\bibitem[Schinzel(2000)]%
        {schinzel2000polynomials}
\bibfield{author}{\bibinfo{person}{Andrzej Schinzel}.}
  \bibinfo{year}{2000}\natexlab{}.
\newblock \bibinfo{booktitle}{\emph{Polynomials with special regard to
  reducibility. With an Appendix by Umberto Zannier.}}
  Vol.~\bibinfo{volume}{77}.
\newblock \bibinfo{publisher}{Cambridge University Press}. 517--x pages.
\newblock


\bibitem[Schmidt(1996)]%
        {schmidt}
\bibfield{author}{\bibinfo{person}{W.~M. Schmidt}.}
  \bibinfo{year}{1996}\natexlab{}.
\newblock \showarticletitle{Heights of points on subvarieties of $\mathbb
  G_m^n$}.
\newblock \bibinfo{journal}{\emph{Number Theory (Paris, 1993– 1994), London
  Math. Soc. Lecture Note Ser. 235}} (\bibinfo{year}{1996}),
  \bibinfo{pages}{157--187}.
\newblock


\bibitem[Tarski(1951)]%
        {tarski1951decision}
\bibfield{author}{\bibinfo{person}{Alfred Tarski}.}
  \bibinfo{year}{1951}\natexlab{}.
\newblock \showarticletitle{A decision method for elementary algebra and
  geometry}.
\newblock  (\bibinfo{year}{1951}).
\newblock


\bibitem[Waldschmidt(2020)]%
        {waldschmidt2020thue}
\bibfield{author}{\bibinfo{person}{Michel Waldschmidt}.}
  \bibinfo{year}{2020}\natexlab{}.
\newblock \showarticletitle{Thue {D}iophantine Equations: A Survey}.
\newblock \bibinfo{journal}{\emph{Class Groups of Number Fields and Related
  Topics}} (\bibinfo{year}{2020}), \bibinfo{pages}{25--41}.
\newblock


\bibitem[Zannier(2012)]%
        {zannier2012some}
\bibfield{author}{\bibinfo{person}{Umberto Zannier}.}
  \bibinfo{year}{2012}\natexlab{}.
\newblock \bibinfo{booktitle}{\emph{Some Problems of Unlikely Intersections in
  Arithmetic and Geometry (AM-181)}}.
\newblock \bibinfo{publisher}{Princeton University Press}.
\newblock


\end{thebibliography}

\appendix
\section{Halfplane Targets}
\label{sec:halfplane targets}
We will assume that $\lambda/\conj\lambda$ is not a root of unity, because this case is essentially the same as the case where the eigenvalues are real. Matrices in which no ratio of distinct eigenvalues is a roots of unity are called \textbf{non-degenerate}.

We begin by noting the first difference between arbitrary dimension and the affine plane, as regards the multiple reachability problem: when the target is a homogeneous hyperplane (in this case a line passing through the origin), it cannot be reached more than once, unless the matrix has a very special form. A consequence of this fact and the work in~\cite{shaullstacs}, which gives an algorithm for deciding single reachability in dimension $2$, is that multiple reachability is decidable for such targets. This is not the case in dimension 10 or higher.
\begin{proposition}
  \label{prop:lines through origin}
  Let $\vp\in\rel^2 \setminus \{(0,0)\}$, $h$ a line going through the origin given by the normal vector $\vec h\in\rel^2$, and $M\in\rel^{2\times 2}$ a non-degenerate matrix. Suppose there are distinct positive integers $n,m\in\nat$ such that both $M^n$ and $M^m$ send $\vp$ to the line $h$, \ie
  \begin{align}
    \label{eq:two zeros}
    \vp\ M^{n}\ \vec h^{\top}=\vp\ M^{m}\ \vec h^{\top}=0.
  \end{align}
  Then $\vp M^k\vec h^\top=0$ for all $k\in\nat$. Moreover, in this case, either one of the eigenvalues of $M$ is zero, or
  \begin{align*}
   M=\begin{pmatrix} s & 0\\ 0 & s\end{pmatrix}, 
  \end{align*}
  for some $s\in\rel$. 
\end{proposition}
\begin{proof}
  By assumption \eqref{eq:two zeros} the point $\vec h$ belongs to the two lines defined by $\vp M^n$ and $\vp M^m$, which pass through the origin. Since $\vec h\ne \vec 0$, it follows that there is some $r\in\rel$, $r\ne 0$, such that
  \begin{align*}
    r\ \vp\ M^n=\vp\ M^m.
  \end{align*}
  If $M$ is not invertible then one of the eigenvalues is $0$, and by putting $M$ into Jordan normal form, we can see that \eqref{eq:two zeros} cannot hold unless $M$ is the zero matrix, or the other eigenvalue is $1$, in which case the conclusion holds. If $M$ is invertible then
  \begin{align*}
    r\ \vp = \vp\ M^{m-n}
  \end{align*}
  and hence $r$ is an eigenvalue of $M^{m-n}$.
  By non-degeneracy, the matrix $M$ has eigenvalue $R:=r^{1/(m-n)}$, which is real. The scaled matrix $\widetilde M=M/R$ has the property that for any $k\in\nat$, $\widetilde M^k$ sends $\vp$ to the line~$h$ if and only if $M^k$ does as well. The matrix $\widetilde M$ has $1$ as an eigenvalue, and for \eqref{eq:two zeros} to hold, $\widetilde M$ (and also $M$) has to be a  stretching matrix, \ie corresponding to multiplication by a scalar $s\in\rel$. Consequently, $\vp \vec h^\top=0$ and hence $\vp M^k\vec h^\top=\vp s^k\vec h^\top=0$ for all $k\in\nat$. 
\end{proof}

The hypothesis that the target line passes through the origin is important. Indeed, perhaps surprisingly, when the target is a line that does \emph{not} pass through the origin, multiple reachability becomes more difficult. What is the difficulty? First, the above proposition fails in that case. Such a target can be reached multiple times.\footnote{There is some work characterising when a line that does not pass through the origin is reached at most once. For example, if the initial point is in $\intg^2$  and the eigenvalue $|\lambda| > 1$, then for all but finitely many such integral initial points the target can be reached at most once \cite{brindza2001multiplicities}.}

Second, almost all known effective methods are based on Baker's work on linear forms in logarithms. Such methods yield an effective time bound, after which it is guaranteed that the orbit will not go in the target. This bound however depends on the height of the initial points. It is not clear how to apply these methods when the initial point is replaced by a set. One possibility is to take the projection of the initial set (as in~\cite{shaullstacs} and the last subsection of this paper), but then the multiple reachability problem is reduced to a problem about intersections of algebraic subgroups with varieties inside tori. There are finiteness results about such intersections, but few of them effective. 

To provide some more intuition, consider a linear map on $\rel^2$. In general, the effect of a linear map on a point consists of (a) a dilation (a shrinking or stretching), and (b) a rotation. When both these effects are relevant, the multiple reachability problem becomes difficult. The positive results that we provide in this section solve decision problems where just one of the effects is at play. For example, the proposition above is about a target that passes through the origin, so the stretching effect of the linear map is not relevant.


A semialgebraic set $\sset S$ is said to be \textbf{bounded} if there exists real $\rho>0$ such that $\sset S$ is contained in the open disk $x^2+y^2<\rho$. We call the infimum among such $\rho$ the \textbf{radius} of the set $\sset S$. The infimum among $\rho\geq0$ such that the set $\sset S$ intersects the open disk of radius $\rho$ is called the \textbf{distance to the origin}. Clearly, boundedness is expressible as a formula in first-order logic, and the radius and distance to the origin are real algebraic by quantifier elimination.

We prove \Cref{thm:pos1}(i), by giving an algorithm that decides multiple reachability for  halfplanes. To this end, let $\sset S$ be the initial semialgebraic set, $\sset T$ the target halfplane, $M$ a $2\times 2$ matrix with rational entries and $m\in\nat$ a positive integer, the minimum number of times we wish to enter the target.   We consider, separately, the case when $M$ has complex conjugate eigenvalues $\lambda, \conj\lambda$, and the case when it has real eigenvalues. We begin with the former.

Let $\vp\in\rel^2$ be a point with polar coordinates $(r,\varphi)$. It is possible to show that there exist real numbers $s, \vartheta, \vartheta_0$ such that for all $n\in\nat$ the polar coordinates of $\vp M^n$ are
\begin{align}
  \label{eq:polar}
  (sr|\lambda|^n, n\vartheta+\vartheta_0+\varphi). 
\end{align}
To see this, simply write $\vp M^n$ as $|\lambda|^n\vp U^n$, where $U$ is a rotation matrix and then follow the second example in the Introduction. The numbers $s, r$ and $|\lambda|$ are real algebraic whose defining formulas (in first-order logic of reals)  can be computed, while $\vartheta$ and $\vartheta_0$ are logarithms of algebraic numbers. We will make use of the following fact from Diophantine approximation. It is a corollary of \cite[Theorem~1 in Page~11]{cassels59_introd_to_dioph_approx}.  For $x \in \rel$, denote by $\fr x$ the unique real number in $[0,2\pi)$ such that, for some integer $m$, $x = 2\pi m + \fr x$. 

\begin{lemma}
  \label{lem:dense}
  If $\vartheta$ is an irrational multiple of $2\pi$, we have
  \begin{align*}
    \set{\fr{n\vartheta} : n\in\nat}\text{is dense in }[0,2\pi]. 
  \end{align*}
\end{lemma}

\begin{proof}[Proof of \Cref{thm:pos1} for non-real eigenvalues]
If $|\lambda|>1$, the algorithm answers \emph{yes}. The justification is as follows. When $\sset T$ is a halfplane, there exist positive real numbers $\alpha_0, \phi_1, \phi_2$, with $\phi_1<\phi_2$,  such that for all $\alpha>\alpha_0$ and $\phi_1<\phi<\phi_2$, the point with polar coordinates $(\alpha, \phi)$ is in $\sset T$. This simply means that the halfplane contains a cone minus a bounded set. 

The matrix $M$ is assumed to be non-degenerate, which implies that the rotation angle $\vartheta$ in \eqref{eq:polar} is an irrational multiple of $2\pi$. So by applying \Cref{lem:dense} to this number, we see that the intersection of the set 
\begin{align}
  \label{eq:vartheta set}
  \set{n\vartheta+\vartheta_0+\phi\mod 2\pi\st n\in\nat}
\end{align}
and the interval $(\phi_1, \phi_2)$ contains infinitely many points. From $|\lambda|>1$, it follows that the sequence of points $\vp M^n$ will enter the cone mentioned above, which is a subset of $\sset T$, infinitely many times.

Suppose now that $|\lambda|<1$.\footnote{The rotation case $|\lambda|=1$ is handled in the next subsection in a more general setting. } When the halfplane $\sset T$ has distance to the origin equal to 0, or when the source $\sset S$ is unbounded, the algorithm answers \emph{yes}, with a justification symmetric to the one above. Assume that $\sset T$ has distance to the origin equal to $\delta>0$ and let  $\sset S$ be bounded with radius $\rho$. Choose some $N\in\nat$  such that $\rho|\lambda|^N<\delta$, then for any source point $\vp\in\sset S$, and all $n>N$, $\vp M^n$ is not in the target $\sset T$. To decide the multiple reachability problem, consider the semialgebraic sets, defined for all $n\in\set{0,1,\ldots, N}$ as
\begin{align*}
  \sset S_n\defeq\set{\vp\in\sset S\st \vp M^n\in\sset T},
\end{align*}
and decide whether there are $m$ among them that have nonempty intersection. 
\end{proof}

We turn our attention now to the case where the eigenvalues of the matrix $M$ are real. We do a case analysis depending on whether the eigenvalues are distinct or not, and whether they are positive or not.

\subsubsection{Diagonalisable $M$ with distinct positive eigenvalues.}
In Jordan normal form, the matrix $M$ is $BDB^{-1}$ where $D$ is the diagonal matrix and $B$ is an invertible matrix with real algebraic entries. We can replace $\sset S$ by $\sset S \cdot B$, and the target set by $B^{-1}\cdot \sset T$. As a consequence we can simply assume that \begin{align*}
	M=\begin{pmatrix}
		\rho_1 & 0\\
		0 & \rho_2
	\end{pmatrix}.
\end{align*}
We will also assume without loss of generality that $\rho_1>\rho_2 > 0$.
The algorithm rests on the following lemma.
\begin{lemma}
	\label{lem:characterisation of orbits}
	Let $M$ be as above, $H$ a halfplane, $\vp\in \mathbb{R}^2$ a point, and $\vp_0,\vp_1,\ldots$ its orbit under $M$.
	The orbit can switch from $H$ to $\rel^2 \setminus H$, or conversely, at most twice. In particular, the orbit is either ultimately in $H$ or ultimately in $\rel^2\setminus H$.
\end{lemma}
\begin{proof}
	We begin by observing that for all real numbers $a_1, a_2, a_3$, not all zero, and positive reals $b_1, b_2$, the function $f\st \rel \to \rel$, defined as
	\begin{align}
		\label{eq:at most two zeros}
		x\mapsto a_1b_1^x+a_2b_2^x+a_3,
	\end{align}
	has at most two zeros. Indeed, since $f$ is continuous, by Rolle's theorem, between any two zeros of $f$, $f'$ has a zero. As a consequence, if $f$ had more than two zeros, $f'$ would have more than one zero. But since $f'$ has the form $\alpha_1b_1^x+\alpha_2b_2^x$ for real numbers $\alpha_1,\alpha_2$, this is impossible.
	
	Let $c_1,c_2,c_3$ be real numbers such that the point $(x,y)$ belongs to the halfplane $H$ if~and~only~if
	\begin{align*}
		c_1x+c_2y+c_3>0.
	\end{align*}
	The orbit of such a point under $M$ is $(x\rho_1^n,y\rho_2^n)$. Consider now the expression
	\begin{align}
		c_1x\rho_1^n + c_2y\rho_2^n+c_3.
		\label{eq:EXP}
	\end{align}
	From the observation about the zeros of \eqref{eq:at most two zeros} above, this expression as a function of $n$ may change sign at most twice, which establishes the lemma.
\end{proof}
From this proof we observe that when the halfplane is given by a homogeneous inequality, the orbit cannot leave the halfplane and come back. For other cases, we proceed to prove that the gaps between consecutive visits to the halfplane $H$ cannot be longer than 3. 

\subsubsection{Diagonalisable $M$ with a single negative eigenvalue.}
Suppose that the matrix $M$ is
\begin{align*}
  M = \begin{pmatrix}\rho_1 & 0\\
                          0 & \rho_2
       \end{pmatrix}
\end{align*}
where $\rho_1 < 0$ and $\rho_2 > 0$. We do not make any assumptions on $|\rho_1|$ and $|\rho_2|$.  Consider a starting point $(x,y) \in \rel^2$ and a halfplane $H$ defined by $c_1x + c_2y > c_3$.  The orbit of $(x,y)$ visits $H$ at time $n$ if

\begin{subnumcases}{}
  c_1x|\rho_1|^n + c_2y\rho_2^n > c_3, & n \text{ even}, \label{eq:eqevens}
  \\
  -c_1x|\rho_1|^n + c_2y\rho_2^n > c_3, & n \text{ odd}. \label{eq:eqodds}
\end{subnumcases}

Depending on the signs of $x$ and $y$, one of the inequalities implies the other.  Without loss of generality suppose \eqref{eq:eqevens} implies \eqref{eq:eqodds}.  By \Cref{lem:characterisation of orbits}, the set of $n$ satisfying \eqref{eq:eqevens} forms an interval in $\nat$.  It follows that the gaps between two consecutive visits from $(x,y)$ to $H$ is at most 2.

\subsubsection{Diagonalisable $M$ with two negative eigenvalues.}
       
Next, suppose that $\rho_1<0$ and $\rho_2<0$.
Clearly, for all $c_1,c_2,c_3\in \mathbb{R}$ with $c_3 \leq 0$ and $c_1,c_2$ not 
both zero, the inequality
$c_1\rho_1^n+c_2\rho_2^n>c_3$ has infinitely many solutions.  
We thus focus
on the case that $c_3>0$.  
Here we prove that the gap between two consecutive visits of the orbit of $(x,y) \in \rel^2$ to $H$ is at most 3. To this end, let $(x,y)\in\rel^2$, and define the function $F\st \rel \to \rel$,
\begin{align*}
  F(t)\defeq c_1x|\rho_1|^t + c_2 y |\rho_2|^t.
\end{align*}
Then we have that for $n\in\nat$,
\begin{align}
  \label{eq:fn or neg}
  c_1x\rho_1^n+c_2 y \rho_2^n =
  \begin{cases}
    \  \ \ F(n) \text{ if $n$ is even,}\\
    -F(n) \text{ if $n$ is odd.}
  \end{cases}
\end{align}
Assuming that $c_1,c_2$ and $x,y$ are nonzero (otherwise we would have an even simpler case), and $\rho_1\ne \rho_2$, we see that the function $F(t)$ is bounded for positive reals $t$ if and only if $|\rho_1|\le 1$ and $|\rho_2|\le 1$. If $F(t)$ is unbounded, then from \eqref{eq:fn or neg} we see that for any $(x,y)\in\rel^2$ nonzero, the system will enter the halfplane $H$ infinitely many times.

If on the other hand $F(t)$ is bounded in $\rel_+$ then the following two inequalities cannot hold simultaneously:
\begin{align*}
  c_1 x \rho_1 + c_2 y \rho_2 &< c_3\\
  c_1 x \rho_1^3 + c_2 y \rho_2^3 &> c_3. 
\end{align*}
Indeed, the two expressions on the left hand side have the same sign, however the second one is smaller in magnitude due to $|\rho_1|\le 1$ and $|\rho_2|\le 1$. The claim that the gaps between two consecutive visits from $(x,y)$ to $H$ is at most 2 follows. 

\subsubsection{Non-diagonalisable $M$ with a repeated eigenvalue.}
       
A version of Lemma~\ref{lem:characterisation of orbits} also holds in case $M$ has a repeated eigenvalue $\rho$.  In this case, every orbit under $M$ can switch from $H$ to $\rel^2 \setminus H$, or conversely, at most once.  Indeed, by a change of basis, we can assume that $M$ has the form
\begin{align*}
  M =\begin{pmatrix}
       \rho & 1 \\
       0 & \rho
     \end{pmatrix}
\end{align*}
Then the expression corresponding to~\eqref{eq:EXP} is
\begin{align*}
  (nxc_2\rho^{-1}+c_2y+c_1x)\rho^n +c_3.
\end{align*}
If $\rho>0$, then it is clear that this expression can change sign at most once as $n$ ranges over $\mathbb{N}$. If, on the other hand, $\rho<0$, we can do a similar analysis as above. If $|\rho|>1$ then the halfplane is entered infinitely often. If $|\rho|\le 1$, we can prove, as we did above, that the gaps between two consecutive visits in $H$ is at most 2. 

\subsubsection{$M$ with a zero eigenvalue.}

This case is one-dimensional, and it can be shown directly that the orbit can switch from $H$ to $\rel^2\setminus H$ (or vice versa) at most once.

Having handled all the cases, we are now ready to give a proof of \Cref{thm:pos1} for real eigenvalues. 

\begin{proof}[Proof of \Cref{thm:pos1}(i) for $M$ with real eigenvalues]\ \linebreak
\Cref{lem:characterisation of orbits} and the case analysis above, implies that any orbit that enters $H$ at least $m$ times must harbour a segment of $m$ visits to $H$ whose gaps between consecutive visits is at most 3. In other words, the orbit of $\vp$ enters $\sset T$ at least $m$ times if and only if there exist $n_1, \ldots, n_m\in\nat$ such that
  \begin{align*}
    \vp M^{n_i}\in\sset T \ \ \text{ and } \:\: 0 < n_{i+1}-n_i \le 3 \:\: \text{ for all $n_i$}. 
  \end{align*}
  This contiguous multiple reachability question can easily be reduced to a union of single reachability queries. Indeed, an orbit contains a pattern (of visits and not visits to $H$) of length $3m$ if and only if it reaches a certain polytope subset $\sset P$ of $\rel^2$;
  A formula defining $P$ can be constructed by considering the sets $\set{x \in \rel^2\st M^k x \in H}$ and $\set{x \in \rel^2\st M^k x \notin H}$ for $0 \le k \le 3m$.
  Thus multiple reachability is reduced to at most $2^{3m}$ instances of single reachability from $\sset S$ to $\sset P$, which can be solved by invoking the algorithm from~\cite{shaullstacs}. 
\end{proof}


\end{document}